\documentclass{myifcolog}

\usepackage{graphicx,soul,latexsym,mathrsfs,graphics,url}
\usepackage{amsfonts,amsmath,amssymb,amsthm,enumerate}
\usepackage{microtype,etex,stmaryrd,booktabs,float,wasysym}
\usepackage[T1]{fontenc}
\usepackage{lmodern}

\newtheorem{theorem}{Theorem}
\newtheorem{example}[theorem]{Example}
\newtheorem{definition}[theorem]{Definition}
\newtheorem{lemma}[theorem]{Lemma}
\newtheorem{fact}[theorem]{Fact}
\newtheorem{proposition}[theorem]{Proposition}

\newtheorem{corollary}[theorem]{Corollary}

\usepackage{hyperref}
\hypersetup{pdfborder = {0 0 0}, breaklinks = true}
\renewcommand{\phi}{\varphi}

\newcommand{\lthen}{\rightarrow}
\newcommand{\Pow}{\mathcal{P}}

\newcommand{\Lng}{\mathcal{L}}

\newcommand{\et}{\wedge}
\newcommand{\Et}{\bigwedge}
\newcommand{\dia}[1]{\langle #1 \rangle}
\newcommand{\union}{\cup}
\newcommand{\Union}{\bigcup}
\newcommand{\then}{\ ; \ }
\newcommand{\ANY}{\mathit{ANY}}
\newcommand{\LNS}{\mathit{LNS}}
\newcommand{\Exp}{\mathit{Ex}}
\newcommand{\gographs}{\mathcal{G}}
\newcommand{\hard}{\blacksquare}
\newcommand{\soft}{\blacklozenge}
\newcommand{\stephard}{\square}
\newcommand{\stepsoft}{\lozenge}

\usepackage{tikz}
\usetikzlibrary{backgrounds,positioning,trees,shapes,arrows,patterns,topaths}
\tikzstyle{call} = [->, line width=1mm, lightgray]
\tikzstyle{callsequencelabel} = [black, font=\footnotesize]
\tikzstyle{indist} = [-, line width=.5mm, darkgray, bend left, dotted]
\tikzset{>=latex}

	\newcommand{\scaleworldexample}{0.5}

	\newcommand{\wordthreeagents}[2]
	{\begin{tikzpicture}[scale=\scaleworldexample]
		\node (a) at (0,0) {$a$};
		\node (b) at (2,0) {$b$};
		\node (c) at (4,0) {$c$};
		\path[#1] (a) edge (b);
		\path[#2] (b) edge (c);
		\end{tikzpicture}}

	\newcommand{\wordthreeagentswithABwithac}[1]
	{\begin{tikzpicture}[scale=\scaleworldexample]
		\node (a) at (0,0) {$a$};
		\node (b) at (2,0) {$b$};
		\node (c) at (4,0) {$c$};
		\path[#1] (b) edge (c);
		\path[<->] (a) edge (b);
		\path[dashed, ->, bend left] (a) edge (c);
		\end{tikzpicture}}

	\newcommand{\wordthreeagentswithABwithAC}[1]
	{\begin{tikzpicture}[scale=\scaleworldexample]
		\node (a) at (0,0) {$a$};
		\node (b) at (2,0) {$b$};
		\node (c) at (4,0) {$c$};
		\path[->, dashed] (b) edge (c);
		\path[->, bend left] (c) edge (b);
		\path[<->] (a) edge (b);
		\path[<->, bend left] (a) edge (c);
		\end{tikzpicture}}

	\newcommand{\wordthreeagentswithABwithBC}[1]
	{\begin{tikzpicture}[scale=\scaleworldexample]
		\node (a) at (0,0) {$a$};
		\node (b) at (2,0) {$b$};
		\node (c) at (4,0) {$c$};
		\path[<->] (c) edge (b);
		\path[<->] (a) edge (b);
		\path[<-, bend right] (a) edge (c);
		\path[->, bend left, dashed] (a) edge (c);
		\end{tikzpicture}}

		\newcommand{\wordthreeagentswithABwithACbis}[1]
	{\begin{tikzpicture}[scale=\scaleworldexample]
		\node (a) at (0,0) {$a$};
		\node (b) at (2,0) {$b$};
		\node (c) at (4,0) {$c$};
		\path[<->] (b) edge (c);
		\path[->, bend left] (a) edge (c);
		\path[<->] (a) edge (b);
		\end{tikzpicture}}

	\newcommand{\wordthreeagentsTOTAL}
	{\begin{tikzpicture}[scale=\scaleworldexample]
		\node (a) at (0,0) {$a$};
		\node (b) at (2,0) {$b$};
		\node (c) at (4,0) {$c$};
		\path[<->] (b) edge (c);
		\path[<->, bend left] (a) edge (c);
		\path[<->] (a) edge (b);
		\end{tikzpicture}}

	\newcommand{\callexample}[4]{
	\draw[call, bend right=#4] (#1) edge node[callsequencelabel, inner sep=0.7mm]
			{$#3$} (#2);}

	\newcommand{\indist}[4]{
	\draw[indist, bend left=#4] (#1) edge node[below, callsequencelabel, inner sep=0.7mm] {$#3$} (#2);}

\addauthor[hans.van-ditmarsch@loria.fr]{Hans van Ditmarsch}{CNRS, LORIA, University of Lorraine, France \& ReLaX, Chennai, India}
\addauthor[malvin@w4eg.eu]{\href{https://orcid.org/0000-0002-2498-5073}{Malvin Gattinger}}{University of Groningen, The Netherlands}
\addauthor[louwe.kuijer@liverpool.ac.uk]{Louwe B.\ Kuijer}{University of Liverpool, United Kingdom}
\addauthor[pere.pardo.v@gmail.com]{\href{https://orcid.org/0000-0003-0181-4658}{Pere Pardo}}{Ruhr-Universit\"{a}t Bochum, Germany}
\title{Strengthening Gossip Protocols using Protocol-Dependent Knowledge}
\titlerunning{Strengthening Gossip Protocols}
\authorrunning{van Ditmarsch et.~al.}
\titlethanks{This work is based on chapter 6 entitled ``Dynamic Gossip''
of Malvin Gattinger's PhD thesis \cite{GattingerThesis2018}.
Malvin Gattinger is the corresponding author, and was affiliated to the
University of Amsterdam during part of this work.
We would like to thank the anonymous IfCoLog referees for their
helpful feedback and suggestions.}

\jreceived{20 April 2018}
\jvolume{6}
\jnumber{1}
\jyear{2019}
\begin{document}
\maketitle

\begin{abstract}
Distributed dynamic gossip is a generalization of the classic telephone problem
in which agents communicate to share secrets, with the additional twist that
also telephone numbers are exchanged to determine who can call whom.
Recent work focused on the success conditions of simple protocols such as
``Learn New Secrets'' ($\LNS$) wherein an agent $a$ may only call another
agent $b$ if $a$ does not know $b$'s secret.
A protocol execution is successful if all agents get to know all secrets.
On partial networks these protocols sometimes fail because they ignore
information available to the agents that would allow for better coordination.
We study how epistemic protocols for dynamic gossip can be strengthened,
using epistemic logic as a simple protocol language with a new operator
for protocol-dependent knowledge.
We provide definitions of different strengthenings and show that they perform
better than $\LNS$, but we also prove that there is no strengthening of $\LNS$
that always terminates successfully.
Together, this gives us a better picture of when and how epistemic coordination
can help in the dynamic gossip problem in particular and distributed systems
in general.
\end{abstract}

\section{Introduction}

The so-called \emph{gossip problem} is a problem about peer-to-peer information
sharing: a number of agents each start with some private information, and the
goal is to share this information among all agents, using only peer-to-peer
communication channels~\cite{Tijdeman1971:TelephoneProblem}.
For example, the agents could be autonomous sensors that need to pool their
individual measurements in order to obtain a joint observation.
Or the agents could be distributed copies of a database that can each be edited
separately, and that need to synchronize with each
other~\cite{EGKM2004:Epidemic,HPPRS2016:GossipDiscovery,Irv2016:GosSec}.

The example that is typically used in the literature, however, is a bit more
frivolous: as the name suggests, the gossip problem is usually represented as a
number of people \emph{gossiping}~\cite{Hedetniemi1988:GossipSurvey,DEPRS2015:DynamicGossip,DEPRS2017:EpistemicGossip}.
This term goes back to the oldest sources on the topic, such as~\cite{BakSho1972:GossipPhone}.
The gossip scenario gives us not only the name of the gossip problem, but also
the names of some of the other concepts that are used: the private information
that an agent starts out with is called that agent's \emph{secret}, the
communication between two agents is called a \emph{telephone call} and an agent $a$
is capable of contacting another agent $b$ if $a$ \emph{knows $b$'s telephone
number}.

These terms should not be taken too literally. Results on the
gossip problem can, in theory, be used by people that literally just want to
exchange gossip by telephone. But we model information exchange in general and
ignore all other social and fun aspects of gossip among humans --- although
these aspects can also be modeled in epistemic logic~\cite{Klein2017:LogDynGossipDEL}.

For our framework, applications where artificial agents need to synchronize
their information are much more likely. For example, recent ideas to improve
cryptocurrencies like bitcoin and other blockchain applications focus on the
peer-to-peer exchange (gossip) happening in such networks~\cite{SoLeZo2016:Spectre}
or even aim to replace blockchains with directed graphs storing the history of
communication~\cite{Baird2017:Hashgraph}.
Epistemic logic can shed new light on the knowledge of agents participating in blockchain
protocols~\cite{HalpernPass2017:KnowBlockchain,BruFluStu2017:LogicBlockchain}.

There are many different sets of rules for the gossip problem~\cite{Hedetniemi1988:GossipSurvey}.
For example, calls may be one-on-one, or may be conference calls. Multiple calls
may take place in parallel, or must happen sequentially. Agents may only be
allowed to exchange one secret per call, or exchange everything they know.
Information may go both ways during a call, or only in one direction.
We consider only the most commonly studied set of rules: calls are one-on-one,
calls are sequential, and the callers exchange all the secrets they know. So if
a call between $a$ and $b$ is followed by a call between $b$ and $c$, then in
the second call agent $b$ will also tell agent $c$ the secret of agent $a$.

The goal of gossip is that every agent knows every secret.
An agent who knows all secrets is called an \emph{expert}, so the goal is to
turn all agents into experts.

The \emph{classical} gossip problem, studied in the 1970s, assumed a total
communication network (anyone could call anyone else from the start), and focused
on optimal call sequences, i.e.\ schedules of calls which spread all the secrets
with a minimum number of calls, which happens to be $2n-4$ for $n \geq 4$
agents~\cite{Tijdeman1971:TelephoneProblem,Hurkens2000:GossipEfficiently}.
Later, this strong assumption on the network of the gossiping agents was dropped,
giving rise to studies on different network topologies (see~\cite{Hedetniemi1988:GossipSurvey}
for a survey), with $2n-3$ calls sufficing for most networks.

Unfortunately, these results about optimal call sequences only show that such
call sequences exist. They do not provide any guidance to the agents about how
to achieve an optimal call sequence. Effectively, these solutions assume a
central scheduler with knowledge of the entire network, who will come up with an
optimal schedule of calls, to be sent to the agents, who will eventually execute
it in the correct order.
Most results also rely upon synchrony so that agents can execute their calls at
the appropriate time (i.e.\ after some calls have been made, and before some
other calls are made).

The requirement that there be a central scheduler that tells the agents exactly
what to do, is against the spirit of the peer-to-peer communication that we want
to achieve. Computer science has shifted towards the study of \emph{distributed
algorithms} for the gossip problem~\cite{BLL1999:DiscovDistrib,KSSV2000:RandomRumor}.
Indeed, the gossip problem becomes more natural without a central scheduler;
the gossiping agents try to do their best with the information they have when
deciding whom to call.
Unfortunately, this can lead to sequences of calls that are redundant because
they contain many calls that are uninformative in the sense that neither agent
learns a new secret.
Additionally, the algorithm may fail, i.e., it may deadlock, get stuck in a
loop or terminate before all information has been exchanged.

For many applications it is not realistic to assume that every agent
is capable of contacting every other agent. So we assume that every agent has a
set of agents of which they ``know the telephone number'', their neighbors, so
to say, and that they are therefore able to contact.
We represent this as a directed graph, with an edge from agent $a$ to agent $b$
if $a$ is capable of calling $b$.

In classical studies, this graph is typically considered to be unchanging. In
more recent work on \emph{dynamic gossip} the agents exchange both the secrets
and the numbers of their contacts, therefore increasing the connectivity of the
network~\cite{DEPRS2015:DynamicGossip}. We focus on dynamic gossip. In
distributed protocols for dynamic gossip all agents decide on their own whom to
call, depending on their current information~\cite{DEPRS2015:DynamicGossip}, or
also depending on the expectation for knowledge growth resulting from the
call~\cite{DEPRS2017:EpistemicGossip}. The latter requires agents to represent
each other's knowledge, and thus epistemic logic.

Different protocols for dynamic gossip are successful in different classes of
gossip networks. The main challenge in designing such a protocol is to find a
good level of redundancy: we do not want superfluous calls, but the less
redundant a gossip protocol, the easier it fails in particular networks. Another
challenge is to keep the protocol simple. After all, a protocol that requires
the agents to solve a computationally hard problem every time they have to
decide whom to call next, would not be practical.
There is also a trade-off between the content of the message of which a call
consists, and the expected duration of gossip protocols. A nice example of that
is~\cite{HerMaf2017:ShareGossiping}, wherein the minimum number of calls to
achieve the epistemic goal is reduced from quadratic to linear order, however at
the price of more `expensive' messages, not only exchanging secrets but also
knowledge about secrets.

\bigskip

A well-studied protocol is ``Learn New Secrets'' ($\LNS$), in which agents are
allowed to call someone if and only if they do not know the other's secret.
This protocol excludes redundant calls in which neither participant learns
any new secrets.
As a result of this property, all $\LNS$ call sequences are finite. For small
numbers of agents, it therefore has a shorter expected execution length than
the ``Any Call'' ($\ANY$) protocol that allows arbitrary calls at all times
and thus allows infinite call sequences~\cite{DitKokSto2017:GossipReachExpec}.
Additionally, it is easy for agents to check whom they are allowed to call
when following $\LNS$.
However, $\LNS$ is not always successful. On some graphs it can terminate
unsuccessfully, i.e.\ when some agents do not yet know all secrets.
In particular there are graphs where the outcome depends on how the
agents choose among allowed calls~\cite{DEPRS2015:DynamicGossip}.

\bigskip

Fortunately, it turns out that failure of $\LNS$ can often be avoided with some
forethought by the calling agents. That is, if some of the choices available to
the agents lead to success and other choices to failure, it is often possible
for the agents to determine in advance which choices are the successful ones.
This leads to the idea of \emph{strengthening} a protocol. Suppose that $P$
is a protocol that, depending on the choices of the agents, is sometimes
successful and sometimes unsuccessful. A strengthening of $P$ is an addition
to $P$ that gives the agents guidance on how to choose among the options that
$P$ gives them.

The idea is that such a strengthening can leave good properties of a protocol
intact, while reducing the chance of failure. For example, any strengthening of
$\LNS$ will inherit the property that there are no redundant calls: It will
still be the case that agents only call other agents if they do not know their
secrets.

\bigskip

Let us illustrate this with a small example, also featuring as a running
example in the technical sections (see Figure~\ref{figure:executionThree} on
page~\pageref{figure:executionThree}).
There are three agents $a,b,c$. Agent $a$ knows the number of $b$, and $b$ and $c$ know each other's number.
Calling agents exchange secrets and numbers, which may expand the network, and they apply the $\LNS$ protocol, wherein you may only call other agents if you do not know their secret.
If $a$ calls $b$, it learns the secret of $b$ and the number of $c$.
All different ways to make further calls now result in all three agents knowing
all secrets. If the first call is between $b$ and $c$ (and there are no other
first calls than $ab$, $bc$, and $cb$), they learn each other's secret but no
new number. The only possible next call now is $ab$, after which $a$ and $b$
know all secrets but not $c$. But although $a$ now knows $c$'s number, she is
not permitted to call $c$, as she already learned $c$'s secret by calling $b$.
We are stuck. So, some executions of $\LNS$ on this graph are successful and
others are unsuccessful.
Suppose we now strengthen the $\LNS$ protocol into $\LNS'$ such that $b$
and $c$ have to wait before making a call until they are called by another agent.
This means that $b$ will first receive a call from $a$.
Then all executions of $\LNS'$ are successful on this graph.
In fact, there is only \emph{one} remaining execution: $ab;bc;ac$.
The protocol $\LNS'$ is a \emph{strengthening} of the protocol $\LNS$.

\bigskip

The main contributions of this paper are as follows.
We define what it means that a gossip protocol is common knowledge between all agents.
To that end we propose a logical semantics with an individual knowledge modality for protocol-dependent knowledge.
We then define various strengthenings of gossip protocols, both in the logical syntax and in the semantics.
This includes a strengthening called uniform backward induction, a form of backward induction applied to (imperfect information) gossip protocol execution trees.
We give some general results for strengthenings, but mainly apply our strengthenings to the protocol $\LNS$: we investigate some basic gossip graphs (networks) on which we gradually strengthen $\LNS$ until all its executions are successful on that graph.
However, no such strengthening will work for all gossip graphs. This is proved
by a counterexample consisting of a six-agent gossip graph, that requires fairly
detailed analysis. Some of our results involve the calculation and checking of
large numbers of call sequences. For this we use an implementation in Haskell.

\bigskip

Our paper is structured as follows.
In Section~\ref{sec:ELfDGP} we introduce the basic definitions to describe
gossip graphs and a variant of epistemic logic to be interpreted on them.
In particular, Subsection~\ref{subsec:protodep}
introduces a new operator for protocol-dependent knowledge.
In Section~\ref{sec:strengthening} we define semantic and --- using the new
operator --- syntactic ways to strengthen gossip protocols. We investigate how
successful those strengthenings are and study their behavior under iteration.
Section~\ref{sec:imposs} contains our main result, that strengthening $\LNS$ to
a strongly successful protocol is impossible.
In Section~\ref{sec:generalizations} we wrap up and conclude.
The Appendix describes the Haskell code used to support our results.

\section{Epistemic Logic for Dynamic Gossip Protocols}\label{sec:ELfDGP}

\subsection{Gossip Graphs and Calls}

\emph{Gossip graphs} are used to keep track of who knows which secrets
and which telephone numbers.

\begin{definition}[Gossip Graph]\label{def:ggs}
Given a finite set of agents $A$, a \emph{gossip graph} $G$ is a triple
$(A,N,S)$ where $N$ and $S$ are binary relations on $A$ such that
$I \subseteq S \subseteq N$ where $I$ is the identity relation on $A$.
An \emph{initial gossip graph} is a gossip graph where $S = I$.
We write $N_ab$ for $(a,b) \in N$ and $N_a$ for $\{ b \in A \mid N_ab \}$, and similarly for the relation $S$.
The set of all initial gossip graphs is denoted by $\gographs$.
\end{definition}

The relations model the basic knowledge of the agents.
Agent $a$ \emph{knows the number} of $b$ iff $N_a b$
and $a$ \emph{knows the secret} of $b$ iff $S_a b$.
If we have $N_a b$ and not $S_a b$ we also say that
  $a$ knows the \emph{pure number} of $b$.

\begin{definition}[Possible Call; Call Execution]\label{def:calls}
A \emph{call} is an ordered pair of agents $(a,b) \in A \times A$.
We usually write $ab$ instead of $(a,b)$.
Given a gossip graph $G$, a call $ab$ is \emph{possible} iff $N_a b$.
Given a possible call $ab$, $G^{ab}$ is the graph $(A',N',S')$ such that
  $A':=A$,
  $N'_a := N'_b := N_a \union N_b$,
  $S'_a := S'_b := S_a \union S_b$, and
  $N'_c := N_c$, $S'_c := S_c$ for $c \neq a,b$.
For a sequence of calls $ab;cd;\dots$ we write $\sigma$ or $\tau$.
The empty sequence is $\epsilon$.
A sequence of possible calls is a \emph{possible call sequence}.
We extend the notation $G^{ab}$ to possible call sequences by $G^\epsilon := G$ and $G^{\sigma;ab} := {(G^\sigma)}^{ab}$.
Gossip graph $G^\sigma$ is the result of \emph{executing} $\sigma$ in $G$.
\end{definition}

To visualize gossip graphs we draw $N$ with dashed and $S$ with solid arrows.
When making calls, the property $S \subseteq N$ is preserved,
so we omit the dashed $N$ arrow if there already is a solid $S$ arrow.

\begin{example}\label{example:simpleIntro}
Consider the following initial gossip graph $G$ in which $a$ knows the number
of $b$, and $b$ and $c$ know each other's number and no other numbers are known:
\begin{center}
  \begin{tikzpicture}[node distance=1.5cm,>=latex]
    \node (a) [] {$a$};
    \node (b) [right of=a] {$b$};
    \node (c) [right of=b] {$c$};
    \draw (a) [dashed,->] -- (b);
    \draw (b) [dashed,<->] -- (c);
  \end{tikzpicture}
\end{center}
Suppose that $a$ calls $b$. We obtain the gossip graph $G^{ab}$ in which
$a$ and $b$ know each other's secret and $a$ now also knows the number of $c$:
\begin{center}
  \begin{tikzpicture}[node distance=1.5cm,>=latex]
    \node (a) [] {$a$};
    \node (b) [right of=a] {$b$};
    \node (c) [right of=b] {$c$};
    \draw (a) [<->] -- (b);
    \draw (b) [dashed,<->] -- (c);
    \draw (b) [dashed,<->] -- (c);
    \draw[->,dashed] (a) to[bend left] (c);
  \end{tikzpicture}
\end{center}
\end{example}

\subsection{Logical Language and Protocols}
\label{subsec:language}

We now introduce a logical language which we will interpret on gossip graphs.
Propositional variables $N_ab$ and $S_ab$ stand for ``agent $a$ knows the number of agent $b$'' and ``agent $a$ knows the secret of agent $b$'', and $\top$ is the `always true' proposition.
Definitions~\ref{def:language} and~\ref{def:protocol} are by simultaneous induction, as the language construct $K_a^P \phi$ refers to a protocol $P$.

\begin{definition}[Language]\label{def:language}
We consider the language $\Lng$ defined by
  \[ \begin{array}{lll}
    \phi & ::= & \top \mid N_ab \mid S_ab \mid \neg \phi \mid (\phi \wedge \phi) \mid K_a^P \phi \mid [\pi] \phi \\[0.5em]
    \pi  & ::= & ?\phi \mid ab \mid (\pi \then \pi) \mid (\pi \union \pi) \mid \pi^\ast
  \end{array} \]
where $a,b \in A$.
Members of $\Lng$ of type $\phi$ are \emph{formulas} and those of type $\pi$ are \emph{programs}.
\end{definition}

\begin{definition}[Syntactic protocol]\label{def:protocol}
A \emph{syntactic protocol}
$P$ is a program defined by
  \[ P :=
    {\left(\Union_{a \neq b \in A} \left(? (N_ab \land P_{ab}) ; ab \right)\right)}^\ast ;
    ?\Et_{a \neq b \in A} \neg \left( N_ab \land P_{ab} \right)
  \] where for all $a \neq b \in A$, $P_{ab}\in\Lng$ is a formula.
  This formula is called the \emph{protocol condition} for call $ab$ of protocol $P$.
  The notation $P_{ab}$ means that $a$ and $b$ are designated variables in that formula.
\end{definition}
Other logical connectives and program constructs are defined by abbreviation.
Moreover, $N_a bcd$ stands for $N_ab \et N_ac \et N_ad$, and $N_a B$ for $\Et_{b \in B}N_ab$.
We use analogous abbreviations for the relation $S$.
We write $\Exp_a$ for $S_a A$.
We then say that agent $a$ is an \emph{expert}.
Similarly, we write $\Exp_B$ for $\Et_{b \in B} \Exp_b$, and $\Exp$ for $\Exp_A$: all agents are experts.

Construct $[\pi] \phi$ reads as ``after every execution of program $\pi$, $\phi$ (is true).''
For program modalities, we use the standard definition for diamonds:
$\langle\pi\rangle\phi := \lnot[\pi]\lnot\phi$, and further: $\pi^0 := {?\top}$ and for all $n\in\mathbb{N}$, $\pi^n := \pi^{n-1};\pi$.

Our protocols are \emph{gossip} protocols, but as we define no other, we omit the word `gossip'.
The word `syntactic' in syntactic protocol is to distinguish it from the semantic protocol that will be defined later.
It is also often omitted.

Our new operator $K_a^P \phi$ reads as ``given the protocol $P$, agent $a$ knows that $\phi$''.
Informally, this means that agent $a$ knows that $\phi$ on the assumption that it is common knowledge among the agents that they all use the gossip protocol $P$.
The epistemic dual is defined as $\hat{K}_a^P \phi := \lnot K_a^P \lnot \phi$ and can be read as ``given the protocol $P$, agent $a$ considers it possible that $\phi$.''

We note that the language is well-defined, in particular $K_a^P$.
The only variable parts of a protocol $P$ are the protocol conditions $P_{ab}$.
Hence, given $|A|$ agents, and the requirement that $a \neq b$, a protocol is determined by its $|A| \cdot (|A|-1)$ many protocol conditions.
We can therefore see the construct $K_a^P \phi$ as an operator with input $(|A| \cdot (|A|-1))+1$ objects of type formula (namely all these protocol condition formulas plus the formula $\phi$ in $K_a^P \phi$), and as output a more complex object of type formula (namely $K_a^P \phi$).%
\footnote{Alternatively one could define a \emph{protocol condition function} $f \colon A^2 \rightarrow \Lng$ and proceed as follows.
In the language BNF replace $K_a^P \phi$ by $K_a (\vec{\phi_{ab}}, \phi)$ where $a \neq b$ and $\vec{\phi_{ab}}$ is a vector representing $|A| \cdot (|A|-1)$ arguments, and in the definition of protocol replace $P_{ab}$ by $f(a,b)$.
That way, Definition~\ref{def:language} precedes Definition~\ref{def:protocol} and is no longer simultaneously defined.
Then, when later defining the semantics of $K_a (\vec{\phi_{ab}}, \phi)$, replace all $\phi_{ab}$ by $f(a,b)$.}

Note that this means that all knowledge operators in a call condition $P_{ab}$ of a protocol $P$ must be relative to protocols strictly simpler than $P$.
In particular, the call condition $P_{ab}$ cannot contain the operator $K_a^P$, although it may contain $K_a^{P'}$ where $P'$ is less complex than $P$.
So the language is incapable of describing the ``protocol'' $X$ given by ``$a$ is allowed to call $b$ if and only if $a$ knows, assuming that $X$ is common knowledge, that $b$ does not know $a$'s secret.''
This is intentional; the ``protocol'' $X$ is viciously circular so we do not want our language to be able to represent it.

\begin{example}\label{example:LNS}
The ``Learn New Secrets'' protocol ($\LNS$) is the protocol with protocol conditions $\lnot S_ab$ for all $a \neq b \in A$.
This prescribes that you are allowed to call any agent whose secret you do not yet know (and whose number you already know).
The ``Any Call'' protocol ($\ANY$) is the protocol with protocol conditions $\top$ for all $a \neq b \in A$.
You are allowed to call any agent whose number you know.
\end{example}

\noindent
The standard epistemic modality is defined by abbreviation as $K_a \phi := K^\ANY_a \phi$.

\subsection{Semantics of Protocol-Dependent Knowledge}\label{subsec:protodep}

We now define how to interpret the language $\Lng$ on gossip graphs.
A \emph{gossip state} is a pair $(G,\sigma)$ such that $G$ is an initial gossip graph and $\sigma$ a call sequence possible on $G$ (see Def.~\ref{def:calls}).
We recall that $G$ and $\sigma$ induce the gossip graph $G^\sigma = (A,N^\sigma,S^\sigma)$.
This is called the gossip graph \emph{associated} with gossip state $(G,\sigma)$.
The semantics of $\Lng$ is with respect to a given initial gossip graph $G$, and defined on the set of gossip states $(G,\sigma)$ for all $\sigma$ possible on $G$.
Definitions~\ref{def:SyncEpistRel} and~\ref{def:Semantics} are simultaneously defined.

\begin{definition}[Epistemic Relation]\label{def:SyncEpistRel}
Let an initial gossip graph $G = (A,N,S)$ and a protocol $P$ be given.
We inductively define the \emph{epistemic relation} $\sim_a^P$ for agent $a$ over gossip states $(G,\sigma)$, where $G^\sigma = (A,N^\sigma,S^\sigma)$ are the associated gossip graphs.
\begin{enumerate}
  \item $(G,\epsilon) \sim_a^P (G,\epsilon)$;
  \item if $(G,\sigma) \sim_a^P (G,\tau)$, $N^\sigma_b = N^\tau_b$, $S^\sigma_b = S^\tau_b$, and $ab$ is $P$-permitted at $(G,\sigma)$ and at $(G,\tau)$, then $(G,\sigma;ab) \sim_a^P (G,\tau;ab)$; \\
    if $(G,\sigma) \sim_a^P (G,\tau)$, $N^\sigma_b = N^\tau_b$, $S^\sigma_b = S^\tau_b$, and $ba$ is $P$-permitted at $(G,\sigma)$ and at $(G,\tau)$, then $(G,\sigma;ba) \sim_a^P (G,\tau;ba)$;
  \item if $(G,\sigma) \sim_a^P (G,\tau)$ and $c,d,e,f \neq a$ such that $cd$ is $P$-permitted at $(G,\sigma)$ and $ef$  is $P$-permitted at $(G,\tau)$,  then $(G,\sigma;cd) \sim_a^P (G,\tau;ef)$.
\end{enumerate}
\end{definition}

\begin{definition}[Semantics]\label{def:Semantics}
Let initial gossip graph $G = (A,N,S)$ be given.
We inductively define the interpretation of a formula $\phi \in \Lng$ on a gossip
state $(G,\sigma)$, where $G^\sigma = (A,N^\sigma,S^\sigma)$ is the associated gossip graph.
\[
\begin{array}{lll}
G,\sigma \models \top          &            & \text{always} \\
G,\sigma \models N_ab          & \text{iff} & N^\sigma_a b \\
G,\sigma \models S_ab          & \text{iff} & S^\sigma_a b \\
G,\sigma \models \neg \phi     & \text{iff} & G,\sigma \not \models \phi \\
G,\sigma \models \phi\land\psi & \text{iff} & G,\sigma \models \phi \text{ and } G,\sigma \models \psi \\
G,\sigma \models K_a^P \phi    & \text{iff} & G,\sigma' \models \phi \text{ for all } (G,\sigma') \sim_a^P (G,\sigma) \\
G,\sigma \models [\pi] \phi    & \text{iff} & G,\sigma' \models \phi \text{ for all } (G,\sigma') \in \llbracket \pi \rrbracket (G,\sigma) \\
\end{array}
\]
where $\llbracket \cdot \rrbracket$ is the following interpretation of programs as relations between gossip states.
Note that we write $\llbracket \pi \rrbracket(G,\sigma)$ for the set $\{(G,\sigma') \mid ((G,\sigma), (G, \sigma')) \in \llbracket \pi \rrbracket \}$.
\[
\begin{array}{rcl}
\llbracket ?\phi \rrbracket {(G,\sigma)}    & := & \{ (G,\sigma) \mid G,\sigma \models \phi \} \\
\llbracket ab \rrbracket {(G,\sigma)}       & := & \{ (G,(\sigma;ab)) \mid G,\sigma \models N_ab \} \\
\llbracket \pi;\pi' \rrbracket {(G,\sigma)} & := & \bigcup \{ \llbracket \pi' \rrbracket {(G,\sigma')} \mid (G,\sigma') \in \llbracket \pi \rrbracket {(G,\sigma)} \} \\
\llbracket \pi\cup\pi' \rrbracket {(G,\sigma)} & := & \llbracket \pi \rrbracket {(G,\sigma)} \cup \llbracket \pi' \rrbracket {(G,\sigma)}\\
\llbracket \pi^\ast \rrbracket {(G,\sigma)}    & := & \bigcup \{ \llbracket \pi^n \rrbracket {(G,\sigma)} \mid n \in \mathbb{N} \} \\
\end{array}
\]
If $G,\sigma \models P_{ab}$ we say that $ab$ is \emph{$P$-permitted} at $(G,\sigma)$.
A \emph{$P$-permitted call sequence} consists of $P$-permitted calls.
\end{definition}

Let us first explain why the interpretation of protocol-dependent knowledge is well-defined.
The interpretation of $K^P_a \phi$ in state $(G,\sigma)$ is a function of the truth of $\phi$ in all $(G,\tau)$ accessible via $\sim_a^P$.
This is standard.
Non-standard is that the relation $\sim_a^P$ is a function of the truth of protocol conditions $P_{ab}$ in gossip states including $(G,\sigma$).
This may seem a slippery slope.
However, note that $K^P_a \phi$ cannot be a subformula of any such $P_{ab}$, as the language $\Lng$ is well-defined: knowledge cannot be self-referential.
These checks of $P_{ab}$ can therefore be performed without vicious circularity.

Let us now explain an important property of $\sim_a^P$, namely that it only relates two gossip states if both are reachable by the protocol $P$.
So if $(G,\sigma)\sim_a^P(G,\sigma')$ and $\sigma$ is a $P$-permitted call sequence, then $\sigma'$ is $P$-permitted as well.
In other words, $a$ assumes that no one will make any calls that are not $P$-permitted.
The set $\{\sim_a^P\mid a\in A\}$ of relations therefore represents the information state of the agents under the assumption that it is common knowledge that the protocol $P$ will be followed.

Given the logical semantics, a convenient primitive is the following \emph{gossip model}.
\begin{definition}[Gossip Model; Execution Tree]\label{def:ModelStatTree}
Given an initial gossip graph $G$, the \emph{gossip model} for $G$ consists of
all \emph{gossip states} $(G,\sigma)$ (where, by definition of gossip states, $\sigma$ is  possible on $G$), with epistemic relations $\sim^P_a$ between gossip states.
The \emph{execution tree} of a protocol $P$ given $G$ is the submodel of the
gossip model restricted to the set of those $(G,\sigma)$ where $\sigma$ is
$P$-permitted.
\end{definition}

The relation $\sim_a^P$ is an equivalence relation on the restriction of a gossip model to the set of gossip states $(G,\sigma)$ where $\sigma$ is $P$-permitted.
This is why we use the symbol $\sim$ for the relation.
However, $\sim_a^P$ is typically not an equivalence relation on the entire domain of the gossip model, as $\sim_a^P$ is not reflexive on unreachable gossip states $(G,\sigma)$.

In our semantics, the modality $[ab]$ can always be evaluated.
There are three cases to distinguish.
$(i)$ If the call $ab$ is not possible (if $a$ does not know the number of $b$), then $\llbracket ab \rrbracket (G,\sigma) = \emptyset$, so that $[ab]\phi$ is trivially true for all $\phi$.
$(ii)$ If the call $ab$ is possible but not $P$-permitted, then $\llbracket ab \rrbracket(G,\sigma) = \{(G,\sigma ; ab)\}$ but $\sim_a^P(G,\sigma;ab) = \emptyset$, so that in such states $K^P_a \bot$ is true: the agent believes everything including contradictions. In other words, we have that $\lnot P_{ab} \lthen [ab] K_c^P \bot$.
$(iii)$ If the call $ab$ is possible and $P$-permitted, then $\llbracket ab \rrbracket(G,\sigma) = \{(G,\sigma ; ab)\}$ and $\sim_a^P(G,\sigma;ab) \neq \emptyset$ consists of the equivalence class of gossip states that are indistinguishable for agent $a$ after call $ab$.

In view of the above, one might want to have a modality or program strictly standing for `call $ab$ is possible and $P$-permitted'. We can enforce protocol $P$ for call $ab$ by $[?P_{ab};ab]\phi$, for ``after the $P$-permitted call $ab$, $\phi$ is true.''

Let us now be exact in what sense the gossip model is a Kripke model. Clear enough, the set of gossip states $(G,\sigma)$ constitute a \emph{domain}, and we can identify the valuation of atomic propositions $N_ab$ (resp.\ $S_ab)$ with the subset of the domain such that $(G,\sigma) \models N_ab$ (resp.\ $(G,\sigma) \models S_ab$).
The relation to the usual accessibility relations of a Kripke model is less clear.
For each agent $a$, we do not have a unique relation $\sim_a$, but parametrized relations $\sim^P_a$; therefore, in a way, there are as many relations for agent $a$ as there are protocols $P$. These relations $\sim^P_a$ are only implicitly given. Given $P$, they can be made explicit if a semantic check of $K^P_a \phi$ so requires.

Gossip models are reminiscent of the history-based models of~\cite{parikhetal:2003} and of the protocol-generated forest of~\cite{jfaketal.JPL:2009}.
A gossip model is a protocol-generated forest (and similarly, the execution trees contained in the gossip model are protocol-generated forests), although a rather small forest, namely consisting of a single tree.
An important consequence of this is that the agents initially have \emph{common knowledge of the gossip graph}.
For example, in the initial gossip graph of the introduction, depicted in  Figure~\ref{figure:executionThree}, agent $a$ knows that agent $c$ only knows the number of $b$.
Other works consider uncertainty about the initial gossip graph (for example, to represent that agent $a$ is uncertain whether $c$ knows $a$'s number), such that each gossip graph initially considered possible generates its own tree \cite{DEPRS2017:EpistemicGossip}.

The gossip states $(G,\sigma)$ that are the domain elements of the gossip model carry along a \emph{history} of prior calls.
This can, in principle, be used in a protocol language to be interpreted on such models, although we do not do this in this work.
An example of such a protocol is the ``Call Once'' protocol described in~\cite{DEPRS2015:DynamicGossip}: call $ab$ is permitted in gossip state $(G,\sigma)$, if $ab$ and $ba$ do not occur in $\sigma$.

With respect to the protocol $\ANY$ the gossip model is not restricted.
If we only were to consider the protocol $\ANY$, to each agent we can associate a unique epistemic relation $\sim^\ANY_a$ in the gossip model, for which we might as well write $\sim_a$. We now have a standard Kripke model. This justifies $K_a \phi$ as a suitable abbreviation of $K^\ANY_a \phi$.

\begin{definition}[Extension of a protocol]\label{def:extension}
For any initial gossip graph $G$ and any syntactic protocol $P$ we define the
\emph{extension of $P$ on $G$} by
\[ \begin{array}{lcl}
P_0(G) &:=& \{ \epsilon \} \\
P_{i+1}(G) &:=&
    \{ \sigma;ab \mid
        \sigma \in P_i(G),
      \ a,b \in A,
      \ G,\sigma \models P_{ab}
    \} \\
P(G) &:=& \bigcup_{i<\omega} P_k(G)
\end{array} \]
The \emph{extension of $P$} is $\{ (G,P(G)) \mid G \in \gographs \}$.
\end{definition}

Recall that $\gographs$ is the set of all initial gossip graphs. We often
identify a protocol with its extension. To compare protocols we will write
$P \subseteq P'$ iff for all $G \in \gographs$ we have $P(G) \subseteq P'(G)$.

\begin{definition}[Success]\label{def-success}
Given an initial gossip graph $G$ and protocol $P$, a $P$-permitted call sequence $\sigma$ is \emph{terminal} iff for all calls $ab$, $G,\sigma \not \models P_{ab}$.
We then also say that the gossip state $(G,\sigma)$ is \emph{terminal}.
A terminal call sequence is \emph{successful} iff after its execution all agents are experts.
Otherwise it is \emph{unsuccessful}.
\begin{itemize}
\item A protocol $P$ is \emph{strongly successful} on $G$ iff all terminal $P$-permitted call sequences are successful: $G, \epsilon \models [P]\Exp$.
\item A protocol is \emph{weakly successful} on $G$ iff some terminal $P$-permitted call sequences are successful: $G, \epsilon \models \dia{P}\Exp$.
\item A protocol is \emph{unsuccessful} on $G$ iff no terminal $P$-permitted call sequences are successful: $G, \epsilon \models [P]\lnot\Exp$.
\end{itemize}
A protocol is \emph{strongly successful} iff it is strongly successful on all initial gossip
graphs $G$, and similarly for weakly successful and unsuccessful.
\end{definition}

Instead of `is successful' we also say `\emph{succeeds}', and instead of `terminal sequence' we also say that the sequence is \emph{terminating}.
Given a gossip graph $G$ and a $P$-permitted sequence $\sigma$ we say that the associated gossip graph $G^\sigma$ is \emph{$P$-reachable} (from $G$).
A terminal $P$-permitted sequence is also called an \emph{execution} of $P$.
Given any set $X$ of call sequences, $\overline{X}$ is the subset of the terminal sequences of $X$.

All our protocols can always be executed.
If this is without making any calls, the protocol extension is empty.
Being empty does not mean that $[P]\bot$ holds, which is never the case.

Strong success implies weak success, but not vice versa.
Formally, we have that $[P]\phi \to \langle P \rangle \phi$ is valid for all
protocols $P$, but $\langle P \rangle \phi \to [P]\phi$ is not valid in general,
because our protocols are typically non-deterministic.

We can distinguish unsuccessful termination (not all agents know all secrets) from successful termination.
In other works~\cite{DEPRS2015:DynamicGossip,AptWoj2017:GossipCK} this distinction cannot be made.
In those works termination implies success.

\begin{example}\label{example:executionThree}
We continue with Example~\ref{example:simpleIntro}.
The execution tree of $\LNS$ on this graph is shown in Figure~\ref{figure:executionThree}.
We denote calls with gray arrows and the epistemic relation with dotted lines.
For example, agent $a$ cannot distinguish whether call $bc$ or $cb$ happened.
At the end of each branch the termination of $\LNS$ is denoted with $\checkmark$ if successful, and $\times$ if unsuccessful.

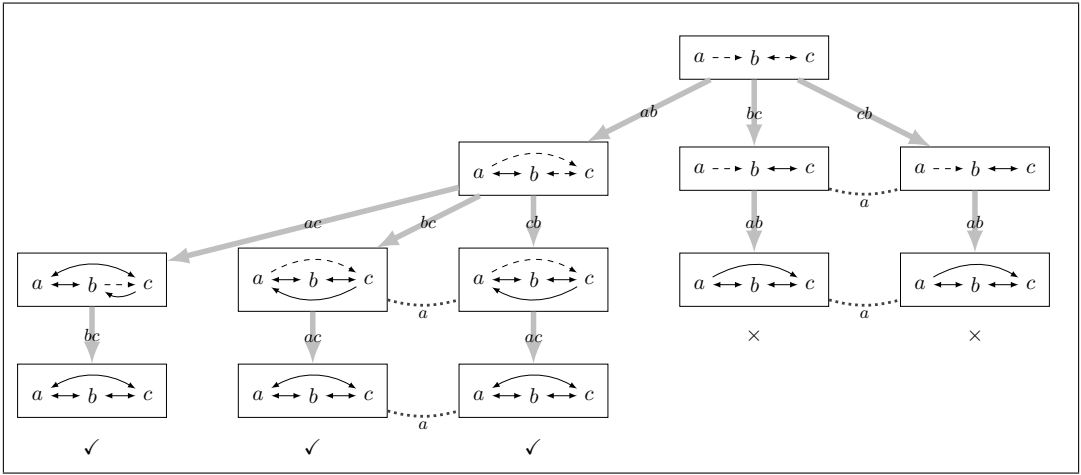
\begin{figure}[ht!]
  \centering
  \tikzstyle{world} = [draw]
  \scalebox{0.73}{\begin{tikzpicture}
    \node (blank1) at (0, .7) {};
    \node (blank2) at (5.5, 0) {};
    \node (blank3) at (0, -7.2) {};
    \node[world] (abbccb) at (0, 0) {\wordthreeagents {->, dashed} {<->, dashed}};
    \node[world] (ABbccb) at (-4, -2) {\wordthreeagentswithABwithac{<->, dashed}};
    \node[world] (BCCB) at (0, -2) {\wordthreeagents{->,dashed} {<->}};
    \node[world] (CBBC) at (4, -2) {\wordthreeagents{->,dashed} {<->}};
    \node[world] (ACABbccb) at (-12, -4) {\wordthreeagentswithABwithAC{<->, dashed}};
    \node[world] (ABBCCB) at (-8, -4) {\wordthreeagentswithABwithBC{<->, dashed}};
    \node[world] (ABCBBC) at (-4, -4) {\wordthreeagentswithABwithBC{<->, dashed}};
    \node[world] (BCCBAB) at (0, -4) {\wordthreeagentswithABwithACbis{->,dashed}};
    \node[world] (CBBCAB) at (4, -4) {\wordthreeagentswithABwithACbis{->,dashed}};
    \node[world] (total1) at (-12, -6) {\wordthreeagentsTOTAL};
    \node[world] (total2) at (-8, -6) {\wordthreeagentsTOTAL};
    \node[world] (total3) at (-4, -6) {\wordthreeagentsTOTAL};
    \callexample{abbccb}{ABbccb}{ab}{0}
    \callexample{abbccb}{BCCB}{bc}{0}
    \callexample{abbccb}{CBBC}{cb}{0}
    \callexample{ABbccb}{ACABbccb}{ac}{0}
    \callexample{ABbccb}{ABBCCB}{bc}{0}
    \callexample{ABbccb}{ABCBBC}{cb}{0}
    \callexample{BCCB}{BCCBAB}{ab}{0}
    \callexample{CBBC}{CBBCAB}{ab}{0}
    \callexample{ACABbccb}{total1}{bc}{0}
    \callexample{ABBCCB}{total2}{ac}{0}
    \callexample{ABCBBC}{total3}{ac}{0}
    \indist{ABBCCB}{ABCBBC}{a}{-15}
    \indist{total2}{total3}{a}{-15}
    \indist{BCCB}{CBBC}{a}{-15}
    \indist{BCCBAB}{CBBCAB}{a}{-15}
    \node (v1) at (-12, -7) {$\checkmark$};
    \node (v2) at (-8, -7) {$\checkmark$};
    \node (v3) at (-4, -7) {$\checkmark$};
    \node (x1) at (0, -5) {$\times$};
    \node (x2) at (4, -5) {$\times$};
  \end{tikzpicture}}
  \caption{Example of an execution tree for $\LNS$.}\label{figure:executionThree}
\end{figure}

To illustrate our semantics, for this graph $G$ we have:
\begin{itemize}
  \item $G,\epsilon \models N_a b \land \lnot S_a b$ ---
    the call $ab$ is $\LNS$-permitted at the start.
  \item $G,\epsilon \models [ab] (S_a b \land S_b a)$ ---
    after the call $ab$ the agents $a$ and $b$ know each other's secret
  \item $G,\epsilon \models [ab] \langle ac \rangle \top$ ---
    after the call $ab$ the call $ac$ is possible.
  \item $G,\epsilon \models [ab] [LNS] \Exp$ ---
    after the call $ab$ the $\LNS$ protocol will always terminate successfully.
  \item $G,\epsilon \models [bc \cup cb] [LNS] \lnot\Exp$ ---
    after the calls $bc$ or $cb$ the $\LNS$ protocol will always terminate unsuccessfully.
  \item $G,\epsilon \models [bc \cup cb] K_a^{LNS} (S_b c \land S_c b)$ ---
    after the calls $bc$ or $cb$, agent $a$ knows that $b$ and $c$ know each
    others secret.
  \item $G,ab;bc;ac \models \bigwedge_{i \in \{ a, b, c \}} K_i^{LNS} \Exp$ ---
    after the call sequence $ab;bc;ac$ everyone knows that everyone is an expert.
\end{itemize}
We only have epistemic edges for agent $a$, and those are between states with identical gossip graphs.
If there are three agents, then if you are not involved in a call, you know that the other two agents must have called.
You may only be uncertain about the direction of that call.
But the direction of the call does not matter for the numbers and secrets being exchanged.
Hence all agents always know what the current gossip graph is.
For a more interesting epistemic relation, see Figure~\ref{figure:nExampleTreePart} in the Appendix.
\end{example}

\subsection{Symmetric and epistemic protocols, and semantic protocols}

Given a protocol $P$, for any $a\neq b$ and $c\not = d$, the protocol conditions $P_{ab}$ and $P_{cd}$ can be different formulas. So a protocol may require different agents to obey different rules.
Although there are settings wherein this is interesting to investigate, we want to restrict our investigation to those protocols where there is one protocol condition to rule them all.
This is enforced by the requirement of \emph{symmetry}.
Another requirement is that the calling agent should know that the protocol condition is satisfied before making a call.
That is the requirement that the protocol be \emph{epistemic}.
It is indispensable in order to see our protocols as \emph{distributed} gossip protocols.

\begin{definition}[Symmetric and epistemic syntactic protocol]\label{def:SynSymEpis}
Let a syntactic protocol $P$ be given. Protocol $P$ is \emph{symmetric} iff
  for every permutation $J$ of agents, we have $\phi_{J(a)J(b)}=J(\phi_{ab})$,
  where $J(\phi_{ab})$ is the natural extension of $J$ to formulas.\footnote{Formally: $J(\top):=\top$, $J(N_ab) := N_ab$, $J(S_ab) := S_ab$, $J(\neg\phi) := \neg J(\phi)$, $J(\phi\et\psi) := J(\phi)\et J(\psi)$,  $J(K_a^P \psi) := K_{J(a)}^{J(P)} J(\psi)$, $J(?\phi) := {? J(\phi)}$, $J(ab) := J(a)J(b)$, $J(\pi;\pi') := J(\pi);J(\pi')$, $J(\pi \cup \pi') := J(\pi) \cup J(\pi')$, $J(\pi^*) := J(\pi)^*$.}
Protocol $P$ is \emph{epistemic} iff for every $a,b\in A$, the protocol condition $P_{ab} \rightarrow K^P_a P_{ab}$ is valid.
We henceforth require all our protocols to be symmetric and epistemic.
\end{definition}

Intuitively, a protocol is \emph{epistemic} if callers always know when to make a call, without being given instructions by a central scheduler.
This means that whenever $P_{ab}$ is true, so agent $a$ is allowed to call agent $b$, it must be the case that $a$ knows that $P_{ab}$ is true.
In other words, in an epistemic protocol $P_{ab}$ implies $K_a^PP_{ab}$.
Furthermore, by Definition~\ref{def:Semantics} knowledge is truthful on the execution tree for protocol $P$ in gossip model.
So except in the gossip states that cannot be reached using the protocol $P$, we also have that $K_a^PP_{ab}$ implies $P_{ab}$.

If a protocol is \emph{symmetric} the names of the agents are irrelevant and therefore interchangeable.
So a symmetric protocol is not allowed to ``hard-code'' agents to perform certain roles.
This means that, for example, we cannot tell agent $a$ to call $b$, as opposed to $c$, just because $b$ comes before $c$ in the alphabet.
But we can tell $a$ to call $b$, as opposed to $c$, on the basis that, say, $a$ knows that $b$ knows five secrets while $c$ only knows two secrets.
If a protocol $P$ is symmetric, we can think of the protocol condition  as the \emph{unique} protocol condition for $P$, modulo permutation.

Epistemic and symmetric protocols capture the distributed peer-to-peer nature
of the gossip problem.

\begin{example}
The protocols $\ANY$ and $\LNS$ are symmetric and epistemic.
For $\ANY$ this is trivial.
For $\LNS$, observe that agents always know which numbers and secrets they know.
A direct consequence of clause (2.) of Definition~\ref{def:SyncEpistRel} of the epistemic relation is that for any protocol $P$, if $(G,\sigma) \sim^P_a (G, \sigma')$, then $N^\sigma_a = N^{\sigma'}_a$ and $S^\sigma_a = S^{\sigma'}_a$.
Thus, applying the clause for knowledge $K^P_a \phi$ of Definition \ref{def:Semantics},
  we immediately get that the following formulas are all valid:
  $N_ab \lthen K_a^P N_ab$,
  $\neg N_ab \lthen K_a^P \neg N_ab$,
  $S_ab \lthen K_a^P S_ab$,
and $\neg S_ab \lthen K_a^P \neg S_ab$.
Therefore, in particular this holds for $P = \LNS$.
\end{example}

Although the numbers and secrets known by an agent before and after a call may vary, the agent always knows \emph{whether} she knows a given number or secret. Knowledge about other agents having a certain number or a secret is preserved after calls. But, of course, knowledge about other agents \emph{not} having a certain number or secret is not preserved after calls.

Not all protocols we discuss in this work are definable in the logical language.
We therefore need the additional notion of a \emph{semantic protocol}, defined by its extension.

\begin{definition}[Semantic protocol]\label{def:semanticprotocol}
A \emph{semantic protocol} is a function
  $P \colon \gographs \to \Pow({(A \times A)}^*)$
mapping initial gossip graphs to sets of call sequences.
We assume semantic protocols to be closed under subsequences, i.e.\ for all $G$
we want that $\sigma;ab \in P(G)$ implies $\sigma \in P(G)$.
For a \emph{semantic protocol} $P$ we say that a call $ab$ is
\emph{$P$-permitted at $(G,\sigma$)} iff $(\sigma;ab) \in P(G)$.
\end{definition}

Given any syntactic protocol we can view its extension as a semantic protocol.
Using this definition of permitted calls for semantic protocols we can apply
Definition~\ref{def:SyncEpistRel} to get the epistemic relation with respect
to a semantic protocol $P$. Because the relation $\sim_a^P$ depends only on
which calls are allowed, the epistemic relation with respect to a (syntactic)
protocol $P$ is identical to the epistemic relation with respect to the
extension of $P$.

We also require that semantic protocols are symmetric and epistemic, adapting
the definitions of these two properties as follows.

\begin{definition}[Symmetric and epistemic semantic protocol]\label{def:SemSymEpis}
A semantic protocol $P$ is \emph{symmetric} iff
  for all initial gossip graphs $G$ and for all permutations $J$ of agents we have $P(J(G)) = J(P(G))$ (where $J(P(G)) := \{ J(\sigma) \mid \sigma \in P(G) \}$).
A semantic protocol $P$ is \emph{epistemic} iff
  for all initial gossip graphs $G$ and for all $\sigma \in P(G)$ we have:
  $(\sigma;ab) \in P(G)$
  iff
  for all $\tau \sim_a^P \sigma$ we have $(\tau;ab) \in P(G)$.
\end{definition}

It is easy to verify that the syntactic definition of an epistemic protocol agrees with the semantic definition.

\begin{proposition}\label{prop:epistemic_agrees}
A syntactic protocol $P$ is epistemic if and only if its extension is epistemic.
\end{proposition}
\begin{proof}
Let $Q$ be the extension of $P$ and note that, as remarked above, the epistemic relations induced by $P$ and $Q$ are identical.
Now we have the following chain of equivalences:
\begin{center}
\begin{tabular}{cl}
  & $P$ is not epistemic \\
  $\Leftrightarrow$ & $\exists a,b,G,\sigma: G,\sigma\not \models P_{ab}\rightarrow K_a^P P_{ab}$\\
  $\Leftrightarrow$ & $\exists a,b,G,\sigma,\tau: G,\sigma\models P_{ab}$, $G,\tau\not \models P_{ab}$ and $(G,\sigma)\sim_a^P (G,\tau)$\\
  $\Leftrightarrow$ & $\exists a,b,G,\sigma,\tau: (\sigma;ab)\in Q(G)$, $(\tau;ab)\not \in Q(G)$ and $(G,\sigma)\sim_a^P (G,\tau)$ \\
  $\Leftrightarrow$ & $\exists a,b,G,\sigma,\tau: (\sigma;ab)\in Q(G)$, $(\tau;ab)\not \in Q(G)$ and $(G,\sigma)\sim_a^{Q} (G,\tau)$ \\
  $\Leftrightarrow$ & $Q$ is not epistemic
\end{tabular}
\end{center}
\vspace{-1em}
\end{proof}

Note that Proposition~\ref{prop:epistemic_agrees} does not imply that every epistemic
semantic protocol is the extension of a syntactic epistemic protocol, since some
semantic protocols are not the extension of any syntactic protocol.

For symmetry, the situation is slightly more complex than for being epistemic.

\begin{proposition}\label{prop:symmetric_agrees}
If a syntactic protocol $P$ is symmetric, then its extension is symmetric.
\end{proposition}
\begin{proof}
Let $Q$ be the extension of $P$. Fix any permutation $J$ and any initial gossip graph $G$.
To show is that $Q(J(G))=J(Q(G))$ (where $J$ is extended to gossip graphs in the natural way).
We show by induction that for every call sequence $\sigma$, we have $\sigma\in Q(J(G)) \Leftrightarrow \sigma \in J(Q(G))$.

As base case, note that $\epsilon\in Q(J(G))$ and $\epsilon \in J(Q(G))$.
Now, as induction hypothesis, assume that for every call sequence $\tau$ that is shorter than $\sigma$, we have $\tau\in Q(J(G)) \Leftrightarrow \tau \in J(Q(G))$.
Let $ab$ be the final call in $\sigma$, so $\sigma = (\tau;ab)$.
Then we have the following sequence of equivalences:
\begin{align*}
  (\tau;ab) \in Q(J(G))
    & \Leftrightarrow J(G),\tau \models P_{ab}\\
    & \Leftrightarrow G,J^{-1}(\tau)\models J^{-1}(P_{ab})\\
    & \Leftrightarrow G,J^{-1}(\tau)\models P_{J^{-1}(ab)}\\
    & \Leftrightarrow (J^{-1}(\tau);J^{-1}(ab))\in Q(G)\\
    & \Leftrightarrow (\tau;ab)\in J(Q(G)),
\end{align*}
where the equivalence on the third line is due to $P$ being symmetric.
This completes the induction step and thereby the proof.
\end{proof}
The converse of Proposition~\ref{prop:symmetric_agrees} does not hold:
if $P$ is not symmetric, it is still possible for its extension to be symmetric.
The reason for this discrepancy is that symmetry for syntactic protocols has the very strong condition that $J(P_{ab})=P_{J(ab)}$.
So if $P$ is symmetric and $P'$ is given by (i) $P'_{cd}=P_{cd}\wedge \top$ and (ii) $P'_{ab}=P_{ab}$ for $a,b \neq c,d$,
  then $P'$ is not symmetric even though $P$ and $P'$ have the same extension.
We do, however, have the following slightly weaker statement.
Recall that a gossip state $(G,\sigma)$ is $P$-reachable iff the call sequence $\sigma$ is $P$-permitted at $G$.
\begin{proposition}
Let $P$ be a syntactic protocol such that, for some $P$-reachable gossip state
$(G,\sigma)$, some permutation $J$ and some $a,b$ we have
  $G,\sigma\not \models P_{J(ab)}\leftrightarrow J(P_{ab})$.
Then the extension of $P$ is not symmetric.
\end{proposition}
\begin{proof}
Let $Q$ be the extension of $P$, and suppose towards a contradiction that $Q$ is symmetric.
Then we have the following sequence of equivalences:
\begin{align*}
G,\sigma\models P_{J(ab)} & \Leftrightarrow (\sigma;J(ab))\in Q(G)\\
& \Leftrightarrow (J^{-1}(\sigma);ab)\in J^{-1}(Q(G))\\
& \Leftrightarrow (J^{-1}(\sigma);ab)\in Q(J^{-1}(G))\\
& \Leftrightarrow J^{-1}(G),J^{-1}(\sigma)\models P_{ab}\\
& \Leftrightarrow G,\sigma\models J(P_{ab}),
\end{align*}
where the equivalence on the third line is due to $Q$ being symmetric.
This contradicts $G,\sigma\not \models P_{J(ab)}\leftrightarrow J(P_{ab})$, from which it follows that $Q$ is not symmetric.
\end{proof}
So while $P$ may be non-symmetric and still have a symmetric extension, this can only happen if $J(P_{ab})$ is equivalent to $P_{J(ab)}$ in all reachable gossip states.
We conclude that our syntactic and semantic definitions of symmetry agree up to logical equivalence.

\section{Strengthening of Protocols}\label{sec:strengthening}

\subsection{How can we strengthen a protocol?}

In our semantics it is common knowledge among the agents that they follow a certain
protocol, for example $\LNS$. Can they use this information to prevent making
``bad'' calls that lead to an unsuccessful sequence?

If we look at the execution graph given in Figure~\ref{figure:executionThree},
then it seems easy to fix the protocol. Agents $b$ and $c$ should wait and not
make the first call. Agent $b$ should not make a call before he has received a
call from $a$. We cannot say this in our logic as we have no converse modalities
to reason over past calls. In this case however, there is a different way to
ensure the same result. We can ensure that $b$ and $c$ wait before calling by a
strengthening of $\LNS$ that only allows a first call from $i$ to $j$ if $j$
does not know the number of $i$. To determine that a call is not the first call,
we need another property: after at least one call happened, there is an agent
who knows another agent's secret.

We can define this new protocol by protocol condition
  $P_{ij} := \LNS_{ij} \land ( \lnot N_j i \lor \bigvee_{k \neq l} S_k l)$.
Observe that this new protocol is again symmetric and epistemic:
agents always know whether $(\lnot N_j i \lor \bigvee_{k \neq l} S_k l)$.
Because of synchronicity, not only the callers but also all other agents know
that there are agents $k$ and $l$ such that $k$ knows the secret of $l$.
This is an ad-hoc solution specific to this initial gossip graph. Could we
also give a general definition to improve $\LNS$ which works on more or even all
initial graphs? The answer to that is: more, yes, but all, no.

We will now discuss different ways to improve protocols by making them more
restrictive. Our goal is to rule out unsuccessful sequences while keeping at
least some successful ones.
Doing this can be difficult because we still require the strengthened protocols to
be epistemic and symmetric. Hence we are not allowed to
arbitrarily rule out specific calls using the names of agents, for example.
Whenever a call is removed from the protocol, we also have to remove all calls
to other agents that the caller cannot distinguish: it has to be done \emph{uniformly}.
But before we discuss specific ideas for strengthening, let us define it.

\begin{definition}[Strengthening]
A protocol $P'$ is a \emph{syntactic strengthening} of a protocol $P$ iff
  $P_{ab}' \rightarrow P_{ab}$ is valid for all agents $a \neq b$.
A protocol $P'$ is a \emph{semantic strengthening} of a protocol $P$ iff
  $P' \subseteq P$.

A \emph{syntactic strengthening procedure} is a function $\heartsuit$ that for
any syntactic protocol $P$ returns a syntactic strengthening $P^\heartsuit$ of $P$.
Analogously, we define \emph{semantic strengthening procedure}.
\end{definition}

We stress that strengthening is a relation between two protocols $P$ and $P'$
whereas strengthening procedures define a restricting transformation that
given any $P$ tells us how to obtain $P'$.
In the case of a syntactic strengthening, $P$ and $P'$ are implicitly required
to be syntactic protocols. Vice versa however, syntactic protocols can be
semantic strengthenings. In fact, we have the following.

\begin{proposition}
Every syntactic strengthening is a semantic strengthening.
\end{proposition}
\begin{proof}
Let $P'$ be a syntactic strengthening of a protocol $P$.
Let a gossip graph $G$ be given.
We show by induction on the length of $\sigma$ that
  $\sigma \in P'(G)$ implies $\sigma \in P(G)$.
The base case where $\sigma=\epsilon$ is trivial.

For the induction step, consider any $\sigma = \tau;ab$.
As $\tau;ab\in P'(G)$, we also have $\tau \in P'(G)$ and $G, \tau \models P'_{ab}$.
From $\tau \in P'(G)$ and the inductive hypothesis, it follows that $\tau \in P(G)$.
From $G, \tau \models P'_{ab}$ and the validity of $P'_{ab} \rightarrow P_{ab}$ follows $G, \tau \models P_{ab}$.
Finally, by Definition~\ref{def:extension}, $\tau\in P(G)$ and $G, \tau \models P_{ab}$ imply $\tau;ab \in P(G)$.
\end{proof}

\begin{lemma}\label{lemma:str}
Suppose $P$ is a strengthening of $Q$.
Then
  $K_a^Q \phi \rightarrow K_a^P \phi$
and
  $\hat{K}_a^P \phi \rightarrow \hat{K}_a^Q \phi$
are both valid, for any agent $a$.
\end{lemma}
\begin{proof}
This follows immediately from the semantics of protocol-dependent knowledge
given in Definition~\ref{def:Semantics}.
\end{proof}

\subsection{Syntactic Strengthening: Look-Ahead and One-Step}

We will now present concrete examples of syntactic strengthening procedures.

\begin{definition}[Look-Ahead and One-Step Strengthenings]\label{def:strengthening}
We define four syntactic strengthening procedures as follows.
Let $P$ be a protocol.
\[ \begin{array}{llll}
\text{hard look-ahead strengthening}: &
P^\hard_{ab} &:=& P_{ab} \land K_a^P [ab] \langle P \rangle \Exp \\
\text{soft look-ahead strengthening}: &
P^\soft_{ab} &:=& P_{ab} \land \hat{K}_a^P [ab] \langle P \rangle \Exp \\
\text{hard one-step strengthening}: &
P^\stephard_{ab} &:=& P_{ab} \land K_a^P [ab] (\Exp \lor \bigvee_{i,j} (N_i j \land P_{ij}) ) \\
\text{soft one-step strengthening}: &
P^\stepsoft_{ab} &:=& P_{ab} \land \hat{K}_a^P [ab] (\Exp \lor \bigvee_{i,j} (N_i j \land P_{ij}) )
\end{array}\]
\end{definition}

The \emph{hard} look-ahead strengthening allows agents to make a call iff the call is
allowed by the original protocol and moreover they \emph{know} that making this call
yields a situation where the original protocol can still succeed.

For example, consider $\LNS^\hard$. Informally, its condition is that $a$ is
permitted to call $b$ iff $a$ does not have the secret of $b$ and $a$ knows
that after making the call to $b$, it is still possible to follow $\LNS$ in
such a way that all agents become experts.

The \emph{soft} look-ahead strengthening allows more calls than the hard look-ahead
strengthening because it only demands that $a$ \emph{considers it possible} that
the protocol can succeed after the call. This can be interpreted as a good faith
or lucky draw assumption that the previous calls between other agents have been
made ``in a good way''. Soft look-ahead strengthening allows agents to take a risk.

The soft and the hard look-ahead strengthening include a diamond $\langle P \rangle$
labeled with the protocol P, where that protocol P by definition contains
arbitrary iteration: the Kleene star $\ast$.
To evaluate this, we need to compute the execution tree of $P$ for the initial gossip graph $G$.
In practice this can make it hard to check the protocol condition of the new protocol.

The \emph{one-step} strengthenings, in contrast, only use the protocol condition
$P_{ij}$ in their formalization and not the entire protocol $P$. This means that
they provide an easier to compute, but less reliable alternative to full
look-ahead, namely by looking only one step ahead.
We only demand that agent $a$ knows (or, in the soft version, considers it
possible) that after the call, everyone is an expert or the protocol can still go
on for at least one more step --- though it might be that all continuation
sequences will eventually be unsuccessful and thus this next call would already
have been excluded by both look-ahead strengthenings.

\bigskip

An obvious question now is, can these or other strengthenings get us from
weak to strong success? Do these strengthenings only remove unsuccessful
sequences, or will they also remove successful branches, and maybe even return
an empty and unsuccessful protocol?
In our next example everything still works fine.

\begin{example}\label{example:exagain}
Consider Example~\ref{example:executionThree} again.
It is easy to see that the soft and the hard look-ahead strengthening rule
out the two unsuccessful branches in this execution tree and keep the successful
ones. Protocol $\LNS^\hard$ only preserves alternatives that are all successful
and $\LNS^\soft$ only eliminates alternatives if they are all unsuccessful. In
the execution tree in Figure~\ref{figure:executionThree}, the effect is the same
for $\LNS^\hard$ and $\LNS^\soft$, because at any state the agents always know
which calls lead to successful branches.
This is typical for gossip scenarios with three agents: if a call happened, the
agent not involved in the call might be unsure about the direction of the call,
but it knows who the callers are.

The one-step strengthenings are not enough to rule out the unsuccessful
sequences. This is because the unsuccessful sequences are of length $2$ but
the one-step strengthenings can only remove the last call in a sequence. In
this case, the protocols $\LNS^\stephard$ and $\LNS^\stepsoft$ rule out
the call $ab$ after $bc$ or $cb$ happened.
\end{example}

\subsection{Semantic Strengthening: Uniform Backward Defoliation}

We now present two semantic strengthening procedures.
They are inspired by the notion of backward induction, a well-known solution concept in decision theory and game theory~\cite{OsbRub1994:GameTheory}.
We will discuss this at greater length when defining the arbitrary iteration of these semantic strengthenings and in Section~\ref{sec:conclusion}.

In backward induction, given a game tree or search tree, a parent node is called \emph{bad} if all its children are loosing or bad nodes.
Similarly, in trees with information sets of indistinguishable nodes, a parent node can be called bad if all its children are bad \emph{and if also all children from indistinguishable nodes are bad}.
Similar notions were considered in~\cite{BalSmeZve2009:KeepHoping,Perea2014:BelFutRat}.
Again, we have a soft and a hard version.
We define \emph{uniform backward defoliation} on the execution trees of dynamic gossip as follows to obtain two semantic strengthenings.
We choose the name ``defoliation'' here because a single application of this strengthening procedure only removes leaves and not whole branches of the execution tree.
The iterated versions we present later are then called \emph{uniform backward induction}.

\begin{definition}[Uniform Backward Defoliation]\label{def-UBR}
Suppose we have a protocol $P$ and an initial gossip graph $G$.
We define the \emph{Hard Uniform Backward Defoliation} $(\mathsf{HUBD})$
and \emph{Soft Uniform Backward Defoliation} $(\mathsf{SUBD})$ of $P$ as follows.
\[ \begin{array}{llll}
  P^{\mathsf{HUBD}}(G) & := &
    \{ \sigma \in P(G) \mid & \sigma = \epsilon, \text{ or } \sigma = \tau;ab \text{ and } \forall (G,\tau') \sim_a^P (G,\tau) \\
    & & & \text{such that } \tau' \in \overline{P(G)} \text{ implies } (G,\tau';ab) \models \Exp \} \\[0.5em]
  P^{\mathsf{SUBD}}(G) & := &
    \{ \sigma \in P(G) \mid & \sigma = \epsilon, \text{ or } \sigma = \tau;ab \text{ and } \exists (G,\tau') \sim_a^P (G,\tau) \\
    & & & \text{such that } \tau' \in \overline{P(G)} \text{ implies } (G,\tau';ab) \models \Exp \}
\end{array} \]
\end{definition}

In this definition, $\forall (G,\tau') \sim_a^P (G,\tau)$ implicitly stands for
  ``for all $\tau' \in P(G)$ such that $(G,\tau') \sim_a^P (G,\tau)$'',
because for $(G,\tau')$ to be in $\sim_a^P$ relation to another gossip state,
$\tau'$ must be $P$-permitted; similarly for the existential quantification.

The \textsf{HUBD} strengthening keeps the calls which \emph{must} lead to a
non-terminal state or a state where everyone is an expert and \textsf{SUBD}
keeps the calls which \emph{might} do so.
Equivalently, we can say that \textsf{HUBD} removes calls which may go wrong
and \textsf{SUBD} removes those calls which will go wrong --- where going
wrong means leading to a terminal node where not everyone is an expert.

We can now prove that for any gossip protocol
  \emph{Hard Uniform Backward Defoliation} is the same as \emph{Hard One-Step Strengthening},
in the sense that their extensions are the same on any gossip graph,
and that \emph{Soft Uniform Backward Defoliation} is the same as \emph{Soft One-Step Strengthening}.

\begin{theorem}\label{thm:ubr-is-one-step}
$P^\stephard = P^{\mathsf{HUBD}}$ and $P^\stepsoft = P^{\mathsf{SUBD}}$
\end{theorem}
\begin{proof}
Note that $\epsilon$ is an element of both sides of both equations.
For any non-empty sequence we have the following chain of equivalences for the
hard versions of UBD and one-step strengthening:
\[ \def\arraystretch{1.7} \begin{array}{llr}
   (\sigma;ab) \in P^\stephard(G) \\
   \Updownarrow \text{by Definition~\ref{def:extension}} \\
   G,\sigma \models P^\stephard_{ab} \\
\Updownarrow \text{by Definition~\ref{def:strengthening}} \\
  G,\sigma \models P_{ab} \land
  K_a^P [ab] \left(\bigvee_{i,j} (N_i j \land P_{ij}) \lor \Exp \right) \\
\Updownarrow \text{by Definition~\ref{def:Semantics}} \\
  (\sigma;ab) \in P(G) \text{ and }
  (G,\sigma) \vDash K_a^P [ab] \left( \bigvee_{i,j} (N_i j \land P_{ij}) \lor \Exp \right) \\
\Updownarrow \text{by Definition~\ref{def:Semantics}} \\
  (\sigma;ab) \in P(G) \text{ and }
  \forall (G,\sigma') \sim_a^P (G,\sigma) :
    (G,\sigma';ab) \models \bigvee_{i,j} (N_i j \land P_{ij}) \lor \Exp \\
\Updownarrow \text{by Definition~\ref{def-success}} \\
  (\sigma;ab) \in P(G) \text{ and }
  \forall (G,\sigma') \sim_a^P (G,\sigma) :
    \sigma';ab \notin \overline{P(G)} \text{ or } (G,\sigma';ab) \models \Exp \\
\Updownarrow \text{by Definition~\ref{def-UBR}} \\
  (\sigma;ab) \in P^{\mathsf{HUBD}}(G) \\
\end{array} \]

And we have a similar chain of equivalences for the soft versions:
\[ \arraycolsep=2pt\def\arraystretch{1.5} \begin{array}{l}
  (\sigma;ab) \in P^\stepsoft(G) \\
  \Updownarrow \text{by Definition~\ref{def:extension}} \\
  G,\sigma \models P^\stepsoft_{ab} \\
  \Updownarrow \text{by Definition~\ref{def:strengthening}} \\
  G,\sigma \models P_{ab} \land \hat{K}_a^P [ab] \left(\bigvee_{i,j} (N_i j \land P_{ij}) \lor \Exp \right) \\
  \Updownarrow \text{by Definition~\ref{def:Semantics}} \\
  (\sigma;ab) \in P(G) \text{ and } (G,\sigma) \models \hat{K}_a^P [ab] \left( \bigvee_{i,j} (N_i j \land P_{ij}) \lor \Exp \right) \\
  \Updownarrow \text{by Definition~\ref{def:Semantics}} \\
  (\sigma;ab) \in P(G) \text{ and }
    \exists (G,\sigma') \sim_a^P (G,\sigma) :
      (G,\sigma';ab) \models \bigvee_{i,j} (N_i j \land P_{ij}) \lor \Exp \\
  \Updownarrow \text{by Definition~\ref{def-success}} \\
  (\sigma;ab) \in P(G) \text{ and }
    \exists (G,\sigma') \sim_a^P (G,\sigma) :
      \sigma';ab \notin \overline{P(G)} \text{ or } (G,\sigma';ab) \models \Exp \\
  \Updownarrow \text{by Definition~\ref{def-UBR}} \\
    (\sigma;ab) \in P^{\mathsf{SUBD}}(G)
\end{array} \]
\end{proof}

Similarly to backward induction in perfect information games \cite{Aumann1995:BIandCKR}, uniform backward defoliation is \emph{rational}, in the sense that it forces an agent to avoid calls leading to unsuccessful sequences.
The strengthening SUBD avoids a call if it always leads to an unsuccessful sequence.
The strengthening HUBD avoids a call if it sometimes leads to a unsuccessful sequence.

\subsection{Iterated Strengthenings}\label{subsec:iterated-with-update}

The syntactic strengthenings we looked at are all defined in terms of the original protocol.
In $P^\hard_{ab} := P_{ab} \land K_a^P[ab] \langle P \rangle \Exp$ the given
$P$ occurs in three places.
Firstly, in the protocol condition $P_{ab}$ requiring that the call is permitted
according to the old protocol $P$ --- this ensures that the new protocol is a
strengthening of the original $P$.
Secondly, as a parameter to the knowledge operator, in $K^P_a$, which means that
agent $a$ knows that everyone followed $P$ (and that this is common knowledge).
Thirdly, in the part $\langle P \rangle$ assuming that after the considered call
everyone will continue to follow protocol $P$ in the future.

Hence we have strengthened the protocol that the agents use and thereby changed
their behavior, but not their assumptions about what protocol other agents follow.
For example, when $P = \LNS$, all agents now act according to $\LNS^\hard$, on
the assumption that all other agents act according to $\LNS$.
This does not mean that agents cannot determine what they know if $\LNS^\hard$
were common knowledge: each agent $a$ can check that knowledge using $K^{\LNS^\hard}_a \phi$.
But this $K^{\LNS^\hard}_a$ modality is not part of the protocol $\LNS^\hard$.
The agents do not use this knowledge to determine whether to make calls.

But why should our agents stop their reasoning here?
It is natural to iterate strengthening procedures and determine whether we can further improve our protocols by also
updating the knowledge of the agents.

For example, consider repeated hard one-step strengthening:
\[ {(P^\stephard)}^\stephard_{ab} = P^\stephard_{ab} \land
  \hat{K}_a^{P^\stephard} [ab] (\Exp \lor
                                \bigvee_{i,j} (N_i j \land P^\stephard_{ij}) ) \]

In this section we investigate iterations and combinations of strengthening procedures.
In particular we investigate various combinations of hard and soft one-step and
look-ahead strengthening, in order to determine how they relate to each other.

\begin{definition}[Strengthening Iteration]\label{def:iteration}
Let $P$ be a syntactic protocol. For any of the four syntactic strengthening procedures
  $\heartsuit \in \{ \hard, \soft, \stephard, \stepsoft \}$,
we define its iteration by adjusting the protocol condition as follows,
which implies $P^{\heartsuit 1} = P^\heartsuit$:
\[ \begin{array}{lll}
  P^{\heartsuit 0}_{ab}     & := & P_{ab} \\
  P^{\heartsuit (k+1)}_{ab} & := & {(P^{\heartsuit k})}^\heartsuit_{ab}
\end{array} \]
Let now $P$ be a semantic protocol, and let $\heartsuit \in \{\mathsf{HUBD},\mathsf{SUBD}\}$.
We define their iteration, for all gossip graphs $G$, by:
\[ \begin{array}{lll}
  P^{\heartsuit{0}}(G)     &:=& P(G) \\
  P^{\heartsuit{(k+1)}}(G) &:=& {(P^{\heartsuit{k}})}^{\heartsuit}(G)
\end{array} \]
\end{definition}

It is easy to check that Theorem~\ref{thm:ubr-is-one-step} generalizes to the
iterated strengthenings as follows.

\begin{corollary}\label{cor:hardsoft}
For any $k \in \mathbb{N}$, we have:
\[ P^{\stephard k} = P^{\mathsf{HUBD}k} \text{ and }
   P^{\stepsoft k} = P^{\mathsf{SUBD}k} \]
\end{corollary}
\begin{proof}
By induction using Theorem~\ref{thm:ubr-is-one-step}.
\end{proof}

\begin{example}
We reconsider Examples~\ref{example:executionThree} and~\ref{example:exagain},
and we recall that $\LNS^\stephard$ and $\LNS^\stepsoft$ rule out the call
$ab$ after $bc$ or $cb$ happened. To eliminate $bc$ and $cb$ as the first call,
we have to iterate one-step strengthening: ${(\LNS^\stephard)}^\stephard$ is
strongly successful on this graph, as well as ${(\LNS^\stepsoft)}^\stepsoft$,
${(\LNS^\stephard)}^\stepsoft$ and ${(\LNS^\stepsoft)}^\stephard$.
\end{example}

\begin{example}\label{example:nExample}
We consider the ``N''-shaped gossip graph shown below.
There are 21 $\LNS$ sequences for this graph, of which $4$ are successful
($\checkmark$) and 17 are unsuccessful ($\times$).

\begin{center}
\begin{tikzpicture}[>=latex,line join=bevel,node distance=15mm,baseline=1]
  \node (1) {\textrm{1}};
  \node (0) [right of=1, node distance=10mm] {\textrm{0}};
  \node (3) [above of=1] {\textrm{3}};
  \node (2) [above of=0] {\textrm{2}};
  \draw [->,dashed] (3) -- (1);
  \draw [->,dashed] (2) -- (0);
  \draw [->,dashed] (3) -- (0);
\end{tikzpicture}
\hspace{1em}
$\begin{array}{ll}
  20;30;01;31    & \times \\
  20;30;31;01    & \times \\
  20;31;10;30    & \times \\
  20;31;30;10    & \times \\
  30;01;20;31    & \times \\
  30;01;31;20    & \times \\
  30;20;01;21;31 & \checkmark \\
\end{array}
\hspace{1em}
\begin{array}{ll}
  30;20;01;31;21 & \checkmark \\
  30;20;21;01;31 & \checkmark \\
  30;20;21;31;01 & \checkmark \\
  30;20;31;01;21 & \times \\
  30;20;31;21;01 & \times \\
  30;31;01;20    & \times \\
  30;31;20;01;21 & \times \\
\end{array}
\hspace{1em}
\begin{array}{ll}
  30;31;20;21;01 & \times \\
  31;10;20;30    & \times \\
  31;10;30;20    & \times \\
  31;20;10;30    & \times \\
  31;20;30;10    & \times \\
  31;30;10;20    & \times \\
  31;30;20;10    & \times \\
\end{array}$
\end{center}

\noindent We can show the call sequences in a more compact way if we only distinguish call sequences up
to the moment when it is decided whether $\LNS$ will succeed.
Formally, consider the set of minimal $\sigma \in \LNS(G)$ such that for all two
terminal $\LNS$-sequences $\tau,\tau' \in \overline{\LNS(G)}$ extending $\sigma$,
we have $G, \tau \models \Exp$ iff $G, \tau' \models \Exp$.
We will use this shortening convention throughout the paper.

\[\begin{array}{ll}
  20 & \times \\
  30;01 &  \times \\
  30;20;01 & \checkmark \\
  30;20;21 & \checkmark \\
  30;20;31 & \times \\
  30;31 &  \times \\
  31 & \times \\
\end{array}\]

It is pretty obvious what the agents should do here: Agent 2 should not make the
first call but let $3$ call $0$ first. The soft look-ahead strengthening works well
on this graph: It disallows all unsuccessful sequences and keeps all successful
ones. For example, after call $30$, agent $2$ considers it possible that call
$30$ happened and in this case the call $20$ can lead to success. Hence the
protocol condition of $\LNS^\soft$ is fulfilled.
The strengthening $\LNS^\soft$ is strongly successful on this graph.

But note that $2$ does not \emph{know} that $20$ can lead to success, because
the first call could have been $31$ as well and for agent $2$ this would be
indistinguishable from $30$.
Therefore the hard look-ahead strengthening is too restrictive here.
In fact, the only call which ${LNS}^\hard$ still allows is $30$ at the beginning.
After that no more calls are allowed by the hard look-ahead strengthening.

A full list showing which call sequences are allowed by which strengthenings of
$\LNS$ for this example is provided in Table~\ref{table:nExampleExtensions}.
``Full'' means that we continue iterating the strengthening until
$P^{\heartsuit{k}}(G) = P^{\heartsuit{(k + 1)}}(G)$ for the given graph $G$.
Such fixpoints of protocol strengthening will be formally introduced in the
next section.
\end{example}

The hard look-ahead strengthening restricts the set of allowed calls based on a
full analysis of the whole execution tree. One might thus expect, that applying
hard look-ahead more than once would not make a difference. However, we have
the following negative results on iterating hard look-ahead strengthening and
the combination of hard look-ahead and hard one-step strengthening.

\begin{fact}\label{fact:hard-idem-fix}
Hard look-ahead strengthening is not idempotent and does not always yield a
fixpoint of hard one-step strengthening:
\begin{enumerate}[(i)]
  \item There exist a graph $G$ and a protocol $P$ for which
    $P^\hard(G) \neq {(P^\hard)}^\hard(G)$.
  \item There exist a graph $G$ and a protocol $P$ for which
    ${(P^\hard)}^\stephard(G) \neq P^\hard(G)$.
\end{enumerate}
\end{fact}
\begin{proof} \
\begin{enumerate}[(i)]
  \item Let $G$ be the ``N'' graph from Example~\ref{example:nExample} and
    consider the protocol $P = \LNS$. Applying hard look-ahead strengthening
    once only allows the first call $30$ and nothing after that call.
    If we now apply hard look-ahead strengthening again we get the empty set:
    $P^\hard(G) \neq {(P^\hard)}^\hard(G) = \varnothing$.
    See also Table~\ref{table:nExampleExtensions}.
  \item The ``diamond'' graph that we will present in Section~\ref{subsec:diamond}
    can serve as an example here. We can show that the inequality holds for this
    graph by exhaustive search, using our Haskell implementation described in the
    Appendix.
    Plain $\LNS$ has $48$ successful and $44$ unsuccessful sequences on this
    graph. Of these, $\LNS^\hard$ still includes $8$ successful and $8$
    unsuccessful sequences. If we now apply hard one-step strengthening, we get
    ${(\LNS^\hard)}^\stephard$ where $4$ of the unsuccessful sequences are removed.
    See also Table~\ref{table:DiamondExampleExtensions} in the Appendix.
    We note that for $P = LNS$ there is no smaller graph to show the inequality.
    This can be checked by manual reasoning or with our implementation.
    \qedhere
\end{enumerate}
\end{proof}

\noindent

Similarly, we can ask whether the soft strengthenings are related to each other,
analogous to Fact~\ref{fact:hard-idem-fix}. We do not know whether there is a
protocol $P$ for which ${(P^\soft)}^\stepsoft \neq P^\soft$ and leave this as
an open question.

Another interesting property that strengthenings can have is \emph{monotonicity}.
Intuitively, a strengthening is monotone iff it preserves the inclusion relation
between extensions of protocols. This property is useful to study the fixpoint
behavior of strengthenings.
We will now define monotonicity formally and then obtain some results for it.

\begin{definition}\label{def:monotone}
A strengthening $\heartsuit$ is called \emph{monotone} iff
  for all protocols $Q$ and $P$ such that $Q \subseteq P$,
  we also have $Q^\heartsuit \subseteq P^\heartsuit$.
\end{definition}

\begin{proposition}[Soft one-step strengthening is monotone]
Let $P$ be a protocol and $Q$ be an arbitrary strengthening of $P$, i.e. $Q \subseteq P$.
Then we also have $Q^\stepsoft \subseteq P^\stepsoft$.
\end{proposition}
\begin{proof}
As $Q$ is a strengthening of $P$, the formula $Q_{ab} \lthen P_{ab}$ is valid.
We want to show that $Q^\stepsoft_{ab} \lthen P^\stepsoft_{ab}$.
Suppose that $G,\sigma \models Q^\stepsoft_{ab}$, i.e.:
\[ G,\sigma \models Q_{ab} \text{ and } G,\sigma \models
  \hat{K}_a^{Q} [ab] ( \Exp \lor \bigvee_{i,j} (N_i j \land Q_{ij}) ) \]
From the first part and the validity of $Q_{ab} \lthen P_{ab}$, we get $G,\sigma \models P_{ab}$.
The second part and the validity of $Q_{ij} \lthen P_{ij}$ give us
  $G,\sigma \models \hat{K}_a^{Q} [ab] (\Exp \lor \bigvee_{i,j} (N_i j \land P_{ij}))$.
From that and Lemma~\ref{lemma:str} it follows that
  $G,\sigma \models \hat{K}_a^P [ab] (\Exp \lor \bigvee_{i,j} (N_i j \land P_{ij}))$.
Combining these, it follows by definition of soft one-step strengthening that
we have $G,\sigma \models P^\stepsoft_{ab}$.
\end{proof}

\begin{proposition}[Both hard strengthenings are not monotone]\label{prop:hard-non-mono}
Let $P$ and $Q$ be protocols. If $Q \subseteq P$, then
$(i)$ $Q^\hard \subseteq P^\hard$ may not hold, and also
$(ii)$ $Q^\stephard \subseteq P^\stephard$ may not hold.
\end{proposition}
\begin{proof}
(i) \emph{Hard one-step strengthening is not monotone}:

Consider the ``spaceship'' graph below with four agents 0, 1, 2 and 3
where 0 and 3 know 1's number, 1 knows 2's number, and 2 knows no numbers.
\begin{center}
  \begin{tikzpicture}[>=latex,line join=bevel,node distance=1.5cm,baseline=1]
    \node (0)  {0};
    \node (1) [right of=0,below of=0,node distance=14mm] {1};
    \node (3) [left of=1,below of=1,node distance=14mm] {3};
    \node (2) [right of=1] {2};
    \draw [->,dashed] (0) -- (1);
    \draw [->,dashed] (3) -- (1);
    \draw [->,dashed] (1) -- (2);
  \end{tikzpicture}
\end{center}
On this graph the $\LNS$ sequences up to decision point are:
\[\begin{array}{ll}
  01;02    & \times     \\
  01;12    & \times     \\
  01;31;02 & \times     \\
\end{array}\hspace{1em}\begin{array}{ll}
  01;31;12 & \checkmark \\
  01;31;32 & \checkmark \\
  12       & \times     \\
\end{array}\hspace{1em}\begin{array}{ll}
  31;01;02 & \checkmark \\
  31;01;12 & \checkmark \\
  31;01;32 & \times     \\
\end{array}\hspace{1em}\begin{array}{ll}
  31;12    & \times     \\
  31;32    & \times     \\
  \  \\
\end{array}\]
Note that
\[ \LNS^\soft(G) = \left \{ \hspace{-3.2pt} \begin{array}{l}
  (01;31;12;02;32), (01;31;12;32;02), (01;31;32;02;12), \\
  (01;31;32;12;02), (31;01;02;12;32), (31;01;02;32;12), \\
  (31;01;12;02;32), (31;01;12;32;02)
  \end{array} \hspace{-3.2pt} \right \} \]
is strongly successful and therefore hard one-step strengthening does not
change it --- we have ${(\LNS^\soft)}^\stephard(G) = \LNS^\soft(G)$.
On the other hand, consider
\[ \LNS^\stephard(G) = \left \{ \hspace{-3pt}
  \begin{array}{l}
  (01;02;12), (01;12;02), (01;31;02;12), (01;31;02;32), \\
  (01;31;12;32;02), (01;31;32;12;02), (12;01), (12;31), \\
  (31;01;02;12;32), (31;01;12;02;32), (31;01;32;02), \\
  (31;01;32;12), (31;12;32), (31;32;12)
  \end{array} \hspace{-3pt} \right \}
\]
and note that this is not a superset of ${(\LNS^\soft)}^\stephard(G) = \LNS^\soft(G)$,
because we have $(01;31;12;02;32) \in {(\LNS^\soft)}^\stephard(G) = \LNS^\soft(G)$
but $(01;31;12;02;32) \notin \LNS^\stephard(G)$.

Together, we have $\LNS^\soft(G) \subseteq \LNS(G)$ but
${(\LNS^\soft)}^\stephard(G) \not\subseteq \LNS^\stephard(G)$.

Hence $Q = \LNS^\soft \subseteq \LNS = P$ is a counterexample and $\stephard$ is not monotone.

\bigskip

\noindent (ii) \emph{Hard look-ahead strengthening is not monotone}:

For hard look-ahead strengthening we can use the same example.
Because $\LNS^\soft$ is strongly successful, hard look-ahead strengthening does not change it:
  ${(\LNS^\soft)}^\hard(G) = \LNS^\soft(G)$.

Moreover, $\LNS^\hard(G) = \{ (01), (31) \}$ is not a superset of
${(\LNS^\soft)}^\hard(G) = \LNS^\soft(G)$.

Together we have $\LNS^\soft(G) \subseteq \LNS(G)$ but
  ${(\LNS^\soft)}^\hard(G) \not\subseteq \LNS^\hard(G)$,
hence hard look-ahead strengthening is not monotone either.
\end{proof}

This result is relevant for our pursuit to pin down how rational agents can employ common knowledge of a protocol to improve upon it.
It shows that hard look-ahead strengthening is not rational, as follows.

We consider again the ``spaceship'' graph in the proof of Proposition~\ref{prop:hard-non-mono}.
Let us define a \emph{bad call} as a call after which no successful continuation is possible.
Correspondingly, a \emph{good call} is one after which success is still possible.
The initial call could be $12$, but that is a bad call.
All successful $\LNS$ sequences on this graph start with $01;31$ or $31;01$.

Let us place ourselves in the position of agent 3 after the call $01$ has been made.
As far as 3 can tell (if the only background common knowledge is that everyone
follows $\LNS$), the first call may have been 12, at which point no agent can
make a good call because no continuation is successful.
In particular, the second call 31 is then bad.
So 3 will not call 1, because it is possible that the call $31$ is bad, and we are following hard look-ahead.

Symmetrically, the same reasoning is made by agent 0: even if the first call is
$31$, it could also have been $12$, after which any continuation is unsuccessful,
and therefore 0 will not call 1, which again seems irrational.

So nobody will make a call. The extension of $\LNS^\hard$ on this graph is empty.

But as all agents know that $12$ is bad, agent 1 knows this in particular, and
as agent 1 is rational herself, she would therefore not have made that call.
And agents 3 and 0 can draw that conclusion too.
It therefore seems after all irrational for 3 not to call 1, or for 0 not to call 1.

This shows that hard look-ahead strengthening is not rational.
In particular, it ignores the rationality of other agents.

\subsection{Limits and Fixpoints of Strengthenings}\label{subsec:fixpoints}

Given the iteration of strengthenings we discussed in the previous section, it
is natural to consider limits and fixpoints of strengthening procedures.
In this subsection we discuss them and give some small results.
A detailed investigation is deferred to future research.

Note that the protocol conditions of all four basic syntactic strengthenings are
conjunctions with the original protocol condition as a conjunct.
Therefore, all these four strengthenings are \emph{non-increasing}:
For all $\heartsuit \in \{ \hard, \soft, \stephard, \stepsoft \}$
  and all protocols $P$, we have
    $P^\heartsuit \subseteq P$.
The same holds, by definition, for semantic strengthenings.
This implies that if, on any gossip graph, we start with a protocol that only
allows finite call sequences, such as $\LNS$, then applying strengthening
repeatedly will eventually lead to a fixpoint. This fixpoint might be the empty
set, or a non-empty set and thereby provide a new protocol.

For other protocols that allow infinite call sequences, such as $\ANY$, we do
not know if this procedure leads to a unique fixpoint and whether fixpoints
are always reached. We therefore distinguish fixpoints from limits.

\begin{definition}[Strengthening Limit; Fixpoint]\label{def:limitProtocol}
Consider any strengthening $\heartsuit$.
The \emph{$\heartsuit$-limit} of a given protocol $P$ is the semantic protocol
$P^{\heartsuit\ast}$ defined as $\bigcap_{k \in \mathbb{N}} P^{\heartsuit{k}}$.
A given protocol $P$ is a \emph{fixpoint} of a strengthening $\heartsuit$
iff $P = P^\heartsuit$.
\end{definition}

\noindent
Note that limit protocols $P^{\heartsuit\ast}$ are \emph{not} in the logical
language, unlike their constituents $P^{\heartsuit k}$. We now define
$P^{\stephard \ast}$ as \emph{Hard Uniform Backward Induction}, and
$P^{\stepsoft \ast}$ as \emph{Soft Uniform Backward Induction}.
Again using induction on Theorem~\ref{thm:ubr-is-one-step}, it follows that
Uniform Backward Induction is the same as arbitrarily often iterated Uniform
Backward Defoliation.

\begin{corollary}\label{cor:ubi}
\[ P^{\stephard \ast} = P^{\mathsf{HUBD}\ast} \text{ and }
   P^{\stepsoft \ast} = P^{\mathsf{SUBD}\ast}. \]
\end{corollary}

\begin{example}
Consider $P=\LNS$.
The number of $\LNS$ calls between $n$ agents is bounded by $\binom{n}{2}= n(n-1)/2$.
The limit $\LNS^{\heartsuit\ast}$ is therefore reached after a finite number of
iterations, and expressible in the gossip protocol language:
$\LNS^{\heartsuit n(n-1)/2} = \LNS^{\heartsuit\ast}$.
\end{example}

As a further observation, the look-ahead strengthenings are not always the limits
of one-step strengthenings. In other words, we do \emph{not} have for all $G$
that $P^{\stephard\ast}(G) = P^\hard(G)$ or that $P^{\stepsoft\ast}(G) = P^\soft(G)$.
Counterexamples are the ``N'' graph from Example~\ref{example:nExample} and the
extension of various strengthenings relating to the example in the upcoming
Section~\ref{subsec:diamond}, as shown in Table~\ref{table:DiamondExampleExtensions}
in the Appendix.

However, we know by the Knaster-Tarski theorem~\cite{Tarski1955:LatticeThm}
that on any gossip graph soft one-step strengthening $\stepsoft$ has
a unique greatest fixpoint, because $\stepsoft$ is monotone and the
lattice we are working in is the powerset of the set of all call sequences and
thereby complete.

\subsection{Detailed Example: the Diamond Gossip Graph}\label{subsec:diamond}

Consider the initial ``diamond'' gossip graph below.

\begin{center}
  \begin{tikzpicture}
    \node (0) at (0, 2) {0};
    \node (1) at (0,-2) {1};
    \node (2) at (-2,0) {2};
    \node (3) at (2,0) {3};
    \draw [->,dashed] (2) -- (0);
    \draw [->,dashed] (2) -- (1);
    \draw [->,dashed] (3) -- (0);
    \draw [->,dashed] (3) -- (1);
  \end{tikzpicture}
\end{center}

There are 92 different terminating sequences of $\LNS$ calls for this initial graph
of which 48 are successful and 44 are unsuccessful.
Also below we give an overview of all sequences.
For brevity we only list them in the compact way, up to the call after which
success has been decided.

\[ \begin{array}{ll}
  20;01 & \times \\
  20;21 & \times \\
  20;30;01 & \checkmark \\
  20;30;21 & \times \\
  20;30;31 & \checkmark \\
  20;31 & \checkmark \\
\end{array}
\hspace{2em}
\begin{array}{ll}
  21;10 & \times \\
  21;20 & \times \\
  21;30 & \checkmark \\
  21;31;10 & \checkmark \\
  21;31;20 & \times \\
  21;31;30 & \checkmark \\
\end{array}
\hspace{2em}
\begin{array}{ll}
  30;01 & \times \\
  30;20;01 & \checkmark \\
  30;20;21 & \checkmark \\
  30;20;31 & \times \\
  30;21 & \checkmark \\
  30;31 & \times \\
\end{array}
\hspace{2em}
\begin{array}{ll}
  31;10 & \times \\
  31;20 & \checkmark \\
  31;21;10 & \checkmark \\
  31;21;20 & \checkmark \\
  31;21;30 & \times \\
  31;30 & \times \\
\end{array} \]

Table~\ref{table:diamondStatistics} shows how many sequences are permitted
by the different strengthenings. Both soft strengthenings rule out no
successful sequences and rule out some unsuccessful sequences.
The hard look-ahead strengthening removes some successful sequences
and rules out the same number of unsuccessful sequences as the soft lookahead
strengthening, but interestingly enough this is a different set.

This demonstrates that Table~\ref{table:diamondStatistics} may be misleading:
the same number of sequences does not imply the same set of sequences. Table~\ref{table:DiamondExampleExtensions} in the Appendix is more detailed and lists sequences. If a further iteration of a strengthening does not change the number and also not the set of sequences, it has the same extension, and is therefore a fixpoint.
For example, Table~\ref{table:DiamondExampleExtensions} shows that $\LNS^{\stepsoft 2}$ and $\LNS^{\stepsoft 3}$ both have $48$ successful and $32$ unsuccessful sequences on the diamond graph.
They also have the same extension, hence $\LNS^{\stepsoft 2}$ is a fixpoint of $\stepsoft$ on this graph.

\begin{table}
  \centering
  \begin{tabular}{lrr}
    Protocol & \# successful & \# unsuccessful \\
    \toprule
    $\LNS$                &  48 &  44 \\
    $\LNS^\hard$          &   8 &   8 \\
    $\LNS^{\hard 2}$      &   0 &   4 \\
    $\LNS^{\hard 3}$      &   0 &   0 \\
    $\LNS^\soft$          &  48 &   8 \\
    $\LNS^{\soft 2}$      &  48 &   8 \\
    $\LNS^{\soft 3}$      &  48 &   8 \\
    $\LNS^\stephard$      &  24 &  36 \\
    $\LNS^{\stephard 2}$  &   8 &  16 \\
    $\LNS^{\stephard 3}$  &   8 &   4 \\
    $\LNS^{\stephard 4}$  &   0 &   4 \\
    $\LNS^{\stephard 5}$  &   0 &   0 \\
    $\LNS^\stepsoft$      &  48 &  36 \\
    $\LNS^{\stepsoft 2}$  &  48 &  32 \\
    $\LNS^{\stepsoft 3}$  &  48 &  32 \\
    ${(\LNS^\stepsoft)}^{\stephard 3}$        &  16 &  0 \\
    ${({(\LNS^\stepsoft)}^\stephard)}^\hard$  &  16 &  0 \\
  \end{tabular}
  \caption{Statistics for the diamond example.}\label{table:diamondStatistics}
\end{table}

Recall that one-step strengthening is uniform backward
defoliation (Theorem~\ref{thm:ubr-is-one-step}) and that the limit of one-step strengthening is uniform
backward induction (Corollary~\ref{cor:ubi}).
Table~\ref{table:diamondStatistics} shows the difference between
the look-ahead strengthenings and the one-step/defoliation strengthenings.
Although on this ``diamond'' graph, the hard strengthenings $\LNS^{\hard k}$
and $\LNS^{\stephard k}$ have the same fixpoint, namely the empty extension for all $k \geq 4$,
the soft strengthenings $\LNS^{\soft k}$ and $\LNS^{\stepsoft k}$ have different
fixpoints. Both are reached when $k=2$.

We now present two strengthenings that are strongly successful on this graph
(only successfully terminating call sequences remain).

Firstly, consider the protocol ${(\LNS^\stepsoft)}^{\stephard 3}$.
Its extension is as follows, see also Tables~\ref{table:diamondStatistics}
and~\ref{table:DiamondExampleExtensions}.
\[ \begin{array}{l}
20;30;01;31;21 \\
20;30;31;01;21 \\
20;31;10;30;21 \\
20;31;30;10;21 \\
\end{array}\hspace{1em}
\begin{array}{l}
21;30;01;31;20 \\
21;30;31;01;20 \\
21;31;10;30;20 \\
21;31;30;10;20 \\
\end{array}\hspace{1em}
\begin{array}{l}
30;20;01;21;31 \\
30;20;21;01;31 \\
30;21;10;20;31 \\
30;21;20;10;31 \\
\end{array}\hspace{1em}
\begin{array}{l}
31;20;01;21;30 \\
31;20;21;01;30 \\
31;21;10;20;30 \\
31;21;20;10;30 \\
\end{array} \]
Its extension has no sequences
with only four calls. There are sequences with redundant second-to-last calls, for
example $10$ in $20;31;30;10;21$.

Secondly, we present a protocol that is strongly successful on this graph and that has no redundant calls. Its description is far more involved than the previous protocol, but the effort seems worthwhile as is shows that: $(i)$ for some initial gossip graphs we can strengthen $\LNS$ up to finding strongly successful as well as optimal extensions; $(ii)$ the hard and soft strengthening procedures described so far merely touch the surface and are not all that goes around, because one can easily show that the following protocol does not correspond to any of those or their iterations.

We first describe it as a semantic protocol, liberally referring to call
histories in our description (which cannot be done in our logical language)
and only then give a formalization using the syntax of our protocol logic.
Consider the following semantic protocol:

\begin{quote}
(1) agent 2 or agent 3 makes a call to either 0 or 1.

(2) the agent among 2 and 3 that did not make a call in step (1) calls either 0 or 1.

(3) the agent $x$ that made the call in step (2) now makes a second call;
    if $x$ called agent 1 before then $x$ now calls 0 and vice versa.

(4) the agent $y$ that made the call in step (1) now makes a second call;
    if $y$ called agent 1 before then $y$ now calls 0 and vice versa.

(5) if the agent $z$ that was called in step (2) is not yet an expert,
    then $z$ calls the last remaining agent whose secret $z$ does not know.
\end{quote}
Now let us explain why this protocol is strongly successful on the ``diamond'' graph, and why it is a strengthening of $\LNS$. There are four possibilities for the first call: 2 may call 0, 2 may call 1, 3 may call 0 or 3 may call 1. These four cases are symmetrical, so let us assume that the first call is 20. The next call will then be made by agent 3, and there are two possibilities: either 3 also calls agent 0, or 3 calls agent 1. The call sequences, and the secrets known by the agents after each call has been made, are shown in the following two tables.

\[\begin{array}{cccccc}
\multicolumn{6}{c}{\text{First case: 2 and 3 call the same agent}}\\
\text{Stage} & \text{Call} & 0 & 1 & 2 & 3\\
(1) & 20 & \{0,2\} & \{1\} & \{0,2\} & \{3\}\\
(2) & 30 & \{0,2,3\} & \{1\} & \{0,2\} & \{0,2,3\}\\
(3) & 31 & \{0,2,3\} & \{0,1,2,3\} & \{0,2\} & \{0,1,2,3\}\\
(4) & 21 & \{0,2,3\} & \{0,1,2,3\} & \{0,1,2,3\} & \{0,1,2,3\}\\
(5) & 01 & \{0,1,2,3\} & \{0,1,2,3\} & \{0,1,2,3\} & \{0,1,2,3\}
\end{array}\]
\[\begin{array}{cccccc}
\multicolumn{6}{c}{\text{Second case: 2 and 3 call different agents}}\\
\text{Stage} & \text{Call} & 0 & 1 & 2 & 3\\
(1) & 20 & \{0,2\} & \{1\} & \{0,2\} & \{3\}\\
(2) & 31 & \{0,2\} & \{1,3\} & \{0,2\} & \{1,3\}\\
(3) & 30 & \{0,1,2,3\} & \{1,3\} & \{0,2\} & \{0,1,2,3\}\\
(4) & 21 & \{0,1,2,3\} & \{0,1,2,3\} & \{0,1,2,3\} & \{0,1,2,3\}
\end{array}\]
Note that all of these calls are possible, in the sense that all callers know the number of the agent they are calling.
Agents 2 and 3 start out knowing the numbers of 0 and 1, so the calls 20, 30, 21 and 31 are possible from the start.
Furthermore, agent 0 learns the number of agent 1 from agent 2 in the first call, so after the call 20 the call 01 is also possible.

In the second case there is no fifth call, since the agent that received the call in step (2) is already an expert after step (4).
As a result, there are no redundant calls in either possible call sequence.
Furthermore, in either case, all agents become experts.
Finally, every call is to an agent whose secret is unknown to the caller before the call.
So, the described protocol is a strongly successful strengthening of $\LNS$.

The two call sequences shown above are possible if the first call is $20$.
There are six other call sequences corresponding to the other three options for the first call.
Overall, the protocol allows the following 8 sequences.
\[ \begin{array}{l}
20;30;31;21;01 \\
20;31;30;21 \\
\end{array}\hspace{1em}
\begin{array}{l}
21;31;30;20;10 \\
21;30;31;20 \\
\end{array}\hspace{1em}
\begin{array}{l}
30;20;21;31;01 \\
30;21;20;31 \\
\end{array}\hspace{1em}
\begin{array}{l}
31;21;20;30;10 \\
31;20;21;30 \\
\end{array}\]
We can also define a syntactic protocol that has the above semantic protocol as its extension.
This syntactic protocol is not particularly elegant, but it illustrates how the logical language can be used to express more complex conditions.
The call condition $P_{ij}$ of this syntactic protocol is of the form $P_{ij}=K_i\psi_{ij}$ (where $K_i$ abbreviates $K_i^{ANY}$, as defined in Section~\ref{subsec:language}).
This guarantees that the protocol is epistemic, because Lemma~\ref{lemma:str} implies that $K_i\psi_{ij}\rightarrow K_i^P K_i\psi_{ij}$ is valid.
The formula $\psi_{ij}$ is a disjunction with the following five disjuncts, one for each of the clauses (1) -- (5) of the protocol as described above.

The formula $\phi_0 := \bigwedge_k\bigwedge_{l\not = k} \neg S_kl$ holds if and only if no calls have taken place yet. Since agents 2 and 3 are the only ones that know the number of another agent, if $\phi_0$ is true then any agent who can make a call is allowed to make that call. So $\phi_0$ is the first disjunct of $\psi_{ij}$, enabling the call in stage (1).

Defining ``exactly one call has been made'' is a bit harder, but we can do it: after the first call, there will be two agents that know two secrets, while everyone else only knows one secret.
So $\phi_1 := \bigvee_{k\not = l}(S_kl \wedge S_lk\wedge \bigwedge_{m\not \in \{k,l\}}\bigwedge_{n\not = m}\neg S_mn)$ holds if and only if exactly one call has been made.
In that case, any agent that is capable of making calls and only knows their own secret is allowed to make a call, so $\phi_1\wedge \bigwedge_{k\not = i}\neg S_ik$ is the second disjunct of $\psi_{ij}$, enabling the call in stage (2).

In stage (3), the second caller is supposed to make another call.
We make a case distinction based on whether the first two calls were to the same agent or to different agents.
If they were to the same agent, then the second caller now knows three different secrets: $\bigvee_{k\not = i}\bigvee_{l\not \in \{i,k\}}S_i{kl}$.
But that holds not only for the agent who made the second call, but also for the agent that received the second call.
The difference between them is that the secret of the receiver of this call is now known by three agents, while the secret of the caller is known by only two: $\bigwedge_{k\not = i}(S_ki\rightarrow \bigwedge_{l\not \in \{i,k\}}\neg S_li)$.

If the first two calls were to different agents, the second caller knows that every agent now knows exactly two secrets: $K_i\bigwedge_k\bigvee_{l\not = k}(S_kl\wedge \bigwedge_{m\not \in \{k,l\}}\neg S_km)$.
This holds for the receiver of the second call as well, but the difference between them is that the number of the receiver is known to an agent who does not know their secret, while the number of the caller is not: $\bigwedge_{k}(N_ki\rightarrow S_ki)$.

In either case, the target of the call should be the unique agent whose number the caller knows but whose secret the caller does not know.
Since calls are always to an agent whose number is known, we only have to stipulate that the target's secret is not known.
So the third disjunct of $\psi_{ij}$ is
\begin{align*} \neg S_ij\wedge  (&(\bigvee_{k\not = i}\bigvee_{l\not \in \{i,k\}}S_i{kl} \wedge \bigwedge_{k\not = i}(S_ki\rightarrow \bigwedge_{l\not \in \{i,k\}}\neg S_li)) \vee \\
& (K_i\bigwedge_k\bigvee_{l\not = k}(S_kl\wedge \bigwedge_{m\not \in \{k,l\}}\neg S_km)\wedge \bigwedge_{k}(N_ki\rightarrow S_ki))),\end{align*}
enabling the call in stage (3).

It is relatively easy to express when the call in stage (4) should happen: before the third call, all agents know that there is no expert yet, while after the third call all agents consider it possible that there is at least one expert.
This can be expressed as $\hat{K}_i\bigvee_{k}\Exp_k$.
It is slightly more difficult to identify the agent who should make the call.
The agent who should make the call, the one who made the call in stage (1), is the only agent who only knows two secrets, and whose number is only known by agents that also know their secret.
So $\neg\bigvee_{k\not = i}\bigvee_{l\not \in \{i,j\}}S_ikl \wedge \bigwedge_{k}(N_ki\rightarrow S_ki)$.
Finally, the person who should be called in this stage is the unique agent of whom the caller knows the number but not the secret.
The fourth disjunct is therefore $\neg S_ij\wedge \hat{K}_i\bigvee_k\Exp_k \wedge \neg\bigvee_{k\not = i}\bigvee_{l\not \in \{i,j\}}S_ikl \wedge \bigwedge_{k}(N_ki\rightarrow S_ki)$.

Finally, the call in stage (5) should only happen if there remains a non-expert agent.
This non-expert considers it possible that all other agents are experts, so the final disjunct of $\psi_{ij}$ is $\neg S_ij \wedge \hat{K}_i\bigwedge_{k\not = i}\Exp_k$.

On the ``diamond'' graph the extension of the syntactic protocol with call condition $P_{ij}$ is the semantic protocol defined above.
Clearly, this protocol is symmetric.
We already showed that the protocol is epistemic as well.

All in all, this gives us the protocol that we were looking for.
Manually verifying the extension of the protocol is somewhat tedious, so we have also checked the extension using the model checking tool described in the Appendix.

\section{An Impossibility Result on Strengthening LNS}\label{sec:imposs}

\subsection{An Impossibility Result}

In this section we will show that there are graphs where
(i) $\LNS$ is weakly successful and
(ii) no epistemic symmetric strengthening of $\LNS$ is strongly successful.
Recall that we assume that the system is synchronous and that the initial
gossip graph is common knowledge.
Without such assumptions it is even easier to obtain such an impossibility
result, a matter that we will address in the final section.

\begin{theorem}\label{thm:StrongImposs}
There is no epistemic symmetric protocol that is a strongly successful
strengthening of $\LNS$ on all graphs.
\end{theorem}
\begin{proof}
Consider the following ``candy'' graph $G$:
\begin{center}
\begin{tikzpicture}[scale=0.6]
\node (0) at (4,4) {0};
\node (1) at (-1,2) {1};
\node (2) at (2,2) {2};
\node (3) at (6,2) {3};
\node (4) at (9,2) {4};
\node (5) at (4,0) {5};

\draw[->,dashed] (0) -- (2);
\draw[->,dashed] (0) -- (3);
\draw[->,dashed] (1) -- (2);
\draw[->,dashed] (5) -- (2);
\draw[->,dashed] (5) -- (3);
\draw[->,dashed] (4) -- (3);
\end{tikzpicture}
\end{center}
$\LNS$ is weakly successful on $G$, but there is no epistemic symmetric protocol
$P$ that is a strengthening of $\LNS$ and that is strongly successful on $G$.

In~\cite{DEPRS2015:DynamicGossip}, it was shown that $\LNS$ is weakly successful on any
graph that is neither a ``bush'' nor a ``double bush''. Since this graph
$G$ is neither a bush nor a double bush, $\LNS$ is weakly successful on
it.
For example, the sequence
\begin{equation*}02;12;53;43;13;03;23;52;42\end{equation*}
is a successful $\LNS$ sequence which makes everyone an expert.
LNS is not strongly successful on this graph, however. For example,
\begin{equation*}02;12;53;43;13;03;52;42\end{equation*}
is an unsuccessful $\LNS$ sequence, because $5$ learns neither the number
nor the secret of $4$ and no further calls are allowed.

Now, suppose towards a contradiction that $P$ is an epistemic symmetric
strengthening of $\LNS$, and that $P$ is strongly successful on $G$.

Before we look at specific calls made by $P$, we consider a general fact.
Recall that knowing a \emph{pure number} means knowing the number of an
agent without knowing their secret. For any gossip graph and
any agent $a$, if no one has $a$'s pure number, then no call sequence will
result in anyone learning $a$'s pure number. After all, in order to learn $a$'s
number, one would have to call or be called by someone who already knows that
number, but in such a call one would also learn $a$'s secret.

In $\LNS$, you are only allowed to call an agent if you have the number but not the
secret of that agent, i.e., if you have their pure number. It follows that if,
in a given gossip graph, no one has $a$'s pure number, then no $\LNS$ sequence on
that graph will contain any calls where $a$ is the receiver.

In the gossip graph $G$ under consideration, agents 0, 1, 4 and 5 are
in the situation that no one else knows their number. So in particular, no one
knows the pure number of any of these agents. It follows that 2 and 3 are the
only possible targets for $\LNS$ calls in this graph.

Now, let us consider the first call according to $P$. This call must target
$2$ or $3$. The calls $12$ and $43$ are bad calls, since they would result in
1 (resp.~4) being unable to make calls or be called, while still not being an
expert.

This means that either 0 or 5 must make the first call. By symmetry, we can
assume without loss of generality that the first call is $02$. This yields the
following situation.

\begin{center}
  \begin{tikzpicture}[scale=0.6]
    \node (0) at (4,4) {0};
    \node (1) at (-1,2) {1};
    \node (2) at (2,2) {2};
    \node (3) at (6,2) {3};
    \node (4) at (9,2) {4};
    \node (5) at (4,0) {5};
    \draw[<->] (0) -- (2);
    \draw[->,dashed] (2) -- (3);
    \draw[->,dashed] (0) -- (3);
    \draw[->,dashed] (1) -- (2);
    \draw[->,dashed] (5) -- (2);
    \draw[->,dashed] (5) -- (3);
    \draw[->,dashed] (4) -- (3);
  \end{tikzpicture}
\end{center}

Now, let us look at the next call.

\begin{itemize}
  \item The sequence $02;43$ is bad, because that would make it impossible
  for 4 to ever become an expert.

  \item Because of the symmetry of $P$, the initial call could have been $03$
  instead of $02$. The sequence $03;12$ is bad, since 1 cannot become an
  expert, so $03;12$ is not allowed by the strongly successful protocol $P$.

  But agent 1 cannot tell the difference between $03$ and $02$, so from the fact
  that $03;12$ is disallowed and that $P$ is epistemic it follows that $02;12$
  is also disallowed.

  \item The sequence $02;03$ is bad, since $0$ will not be able to make any
  call afterwards. Because $0$ can also never be called, this implies that $0$
  will never become an expert.

  \item Consider then the sequence $02;23$. This results in the following diagram.

  \begin{center}
    \begin{tikzpicture}[scale=0.6]
      \node (0) at (4,4) {0};
      \node (1) at (-1,2) {1};
      \node (2) at (2,2) {2};
      \node (3) at (6,2) {3};
      \node (4) at (9,2) {4};
      \node (5) at (4,0) {5};
      \draw[<->] (0) -- (2);
      \draw[<->] (2) -- (3);
      \draw[->] (3) to[bend right] (0);
      \draw[->,dashed] (0) -- (3);
      \draw[->,dashed] (1) -- (2);
      \draw[->,dashed] (5) -- (2);
      \draw[->,dashed] (5) -- (3);
      \draw[->,dashed] (4) -- (3);
    \end{tikzpicture}
  \end{center}

This graph has the following property: it is impossible (in any $\LNS$
sequence) for any agent to get to learn a new pure number. That is, nobody
can learn a new number without also getting to know the secret of that agent:
agents 1, 0, and 4 each know only one pure number, so they cannot teach anyone a
new number, and agent 5 knows two pure numbers (2 and 3), but those agents
already know each other's secrets.

As a result, any call that will become allowed by $\LNS$ in the future is already
allowed now. There are 5 such calls that are currently allowed, namely 12,
52, 53, 03 and 43. Furthermore, of those calls 52 and 53 are mutually exclusive,
since calling 2 will teach 5 the secret of 3, and calling 3 will teach 5 the
secret of 2.

So any continuation of $02;23$ allowed by $\LNS$ can only contain (in any order)
12, 03, 43 and either 52 or 53. Since $P$ is a strengthening of $\LNS$, the same
holds for $P$. But using only those calls, there is no way to teach 3 the secret
of 1: secret 1 can reach agent 2 using the call 12, but in order for the secret
to travel any further we need the call 52. After that call only 03 and 43 are
still allowed (in particular, 53 is ruled out), so the knowledge of secret 1
remains limited to agents 1, 2 and 5.

Since 02;13 cannot be extended to a successful $\LNS$ sequence, 02;13 must be disallowed.

\item Consider the call sequence $02;52$. This gives the following diagram.
  \begin{center}
    \begin{tikzpicture}[scale=0.6]
      \node (0) at (4,4) {0};
      \node (1) at (-1,2) {1};
      \node (2) at (2,2) {2};
      \node (3) at (6,2) {3};
      \node (4) at (9,2) {4};
      \node (5) at (4,0) {5};
      \draw[<->] (0) -- (2);
      \draw[->,dashed] (2) -- (3);
      \draw[->,dashed] (0) -- (3);
      \draw[->,dashed] (1) -- (2);
      \draw[<->] (5) -- (2);
      \draw[->] (5) -- (0);
      \draw[->,dashed] (5) -- (3);
      \draw[->,dashed] (4) -- (3);
    \end{tikzpicture}
  \end{center}
  Note that in this situation, it is impossible for agents 3 and 4 to learn any
  new number without also learning the secrets corresponding to those numbers:
  there is no agent that knows the number of agent 3 and that also knows another
  pure number, and this will remain the case whatever other calls happen.

  This means that agent 3 cannot make any calls, and that agent 4 can make
  exactly one call, to agent 3.

  Suppose now that $02;52$ is extended to a successful $\LNS$ sequence. This
  sequence has to contain the call 43 at some point. This will be the only call
  by agent 4, so in order for the sequence to be successful, agent 3 already
  has to know secret 1 by the time 43 takes place.

  In particular, this means that the call 12 has already happened, and that
  either agent 1 or agent 2 has then called agent 3 to transmit this secret.
  Whichever agent among 1 and 2 makes this call, afterwards they are unable to
  make any more calls. Furthermore, this takes place before the call 43, so
  whatever agent $x \in \{ 1,2 \}$ informs 3 of secret 1 does not learn secret
  4. Since this agent $x$ can neither make another call nor be called, it
  follows that $x$ does not become an expert.

  So $02;52$ is not allowed by $P$ which we assumed to be strongly successful.

\item Finally, consider the call sequence $02;53$. By symmetry, 03 could have
  been the first call as opposed to 02. Furthermore, the same reasoning that
  showed 02;52 to be unsuccessful above can, with an appropriate permutation of
  agents, be used to show that 03;53 is unsuccessful.

  Agent 5 cannot distinguish between the first call 02 and 03 before making the
  call $53$, so if $03;53$ is disallowed then so is $02;53$ because $P$ is
  epistemic.
\end{itemize}

Remember that $02$ is, without loss of generality, the only initial call that
can lead to success.
We have shown that all of the $\LNS$-permitted calls following the initial call
02 (namely, the calls 43, 12, 03, 23, 52 and 53) are disallowed by $P$.
This contradicts $P$ being a strongly successful strengthening of $\LNS$.
\end{proof}

\subsection{Backward Induction and Look-Ahead applied to Candy}

Given this impossibility result, it is natural to wonder what would happen if we
use the syntactic strengthenings from Definition~\ref{def:strengthening}, or
their iterations, on the ``candy'' graph $G$.

All second calls are eliminated by $\LNS^\hard$, because for any two agents $a$
and $b$ we have $G, 02 \models \lnot K^\LNS_a [ab] \langle \LNS \rangle \Exp$.
By symmetry this also holds for the three other possible first calls,
hence $\LNS^\hard$ is unsuccessful on $G$.
However, the first calls \emph{are} still allowed according to $\LNS^\hard$.

There are 9468 $\LNS$-sequences on this graph of which 840 are successful.
Using the implementation discussed in the Appendix we found out that
$\LNS^\soft$, the soft look-ahead strengthening of $\LNS$, is weakly successful
on this graph and allows 840 successful and 112 unsuccessful sequences.

\section{Conclusions, Comparison, and Further Research}\label{sec:generalizations}\label{sec:conclusion}

\paragraph*{Conclusions}
We modeled common knowledge of protocols in the setting of distributed dynamic
gossip. A crucial role is played by the novel notion of protocol-dependent
knowledge. This knowledge is interpreted using an epistemic relation over
states in the execution tree of a gossip protocol in a given gossip graph. As
the execution tree consists of gossip states resulting from calls permitted by
the protocol, this requires a careful semantic framework.

We described various syntactically or semantically definable strengthenings of
gossip protocols, and investigated the combination and iteration of such
strengthenings, in view of strengthening a weakly successful protocol into one
that is strongly successful on all graphs.
In the setting of gossip, a novel notion we used in such strengthenings is that
of uniform backward induction, as a variation on backward induction in search
trees and game trees.

Finally, we proved that for the $\LNS$ protocol, in which agents are only allowed
to call other agents if they do not know their secrets, it is impossible to
define a strengthening that is strongly successful on all graphs.

\paragraph*{Comparison}
As already described at length in the introductory section, our work builds
upon prior work on dynamic distributed
gossip~\cite{DEPRS2015:DynamicGossip,DEPRS2017:EpistemicGossip},
which itself has a prior history in the networks
community~\cite{BLL1999:DiscovDistrib,KSSV2000:RandomRumor,Haeupler2015:SimpleSpread}
and in the logic community~\cite{ADGH2014:KnowledgeGossip,AptGroHoe2015:EpisDistGos}.
Many aspects of gossip may or may not be common knowledge among agents: how many
agents there are, the time of a global clock, the gossip graph, etc. The point
of our result is that even under the strongest such assumptions, one can still
not guarantee that a gossip protocol always terminates successfully.
How common knowledge of agents is affected by gossip protocol execution is
investigated in~\cite{AptWoj2017:GossipCK}: for example, the authors demonstrate
how sender-receiver subgroup common knowledge is obtained (and lost) during calls.
However, they do not study common knowledge of gossip protocols.
We do not know of other work on that topic. Outside the area of
gossip, protocol knowledge has been well investigated in the epistemic logic
community~\cite{hoshi:phd,Wang10:phd,hvdetal.aij:2014}.

While the concept of backward induction is well-known in game theory
(see for example \cite{Aumann1995:BIandCKR}), it is only
used in perfect-information settings, where all agents know what the real world
or the actual state is. Our definition of \emph{uniform} backward induction is a
generalization of backward induction to the dynamic gossip setting, where only
partial observability is assumed. A concept akin to uniform backward induction
has been proposed in~\cite{Perea2014:BelFutRat}
(rooted in~\cite{BattiSini2002:StrongBelFwIR}), under the name of
\emph{common belief in future rationality}, with an accompanying recursive
elimination procedure called \emph{backward dominance}.\footnote{We kindly thank
Andr\'es Perea for his interactions.} As in our approach, this models a decision
rule faced with uncertainty over indistinguishable moves.
In~\cite{Perea2014:BelFutRat}, the players are utility maximizers with
probabilistic beliefs, which in our setting would correspond to
\emph{randomizing} over all indistinguishable moves/calls. As a decision rule
this is also known as the \emph{insufficient reason} (or \emph{Laplace})
criterion: all outcomes are considered equiprobable. Seeing uniform backward
induction as the combination of backward induction and a decision rule
immediately clarifies the picture.
Soft uniform backward induction applies the \emph{minimax regret} criterion for
the decision whom to call, minimizing the maximum utility loss. In contrast,
hard uniform backward induction applies the \emph{maximin utility} criterion,
maximizing the minimum utility (also known as risk-averse, pessimistic, or Wald
criterion).

In the gossip scenario, the unique minimum value is unsuccessful termination,
and the unique maximum value is successful termination.
Minimax prescribes that as long as the agent considers it possible that a call
leads to successful termination, the agent is allowed to make the call (as long
as the minimum of the maximum is success, go for it): the soft version.
Maximin prescribes that, as long as the agent considers it possible that a call
lead to unsuccessful termination, the agent should not make the call (as long as
the maximum of the minimum is failure, avoid it): the hard version.
Such decision criteria over uncertainty also crop up in areas overlapping with
social software and social choice, e.g.~\cite{BalSmeZve2009:KeepHoping,CoWaXi2011:ManVot,parikhetal:2013,Meir2015:PluVotUnc}.
In~\cite{BalSmeZve2009:KeepHoping} a somewhat similar concept has been called
``common knowledge of stable belief in rationality''. However, there it applies
to a weaker epistemic notion, namely belief.

\paragraph*{Further Research}
The impossibility result for $\LNS$ is for dynamic gossip where agents
exchange both secrets and numbers, and where the network expands.
Also in the non-dynamic setting we can quite easily find a graph where static
$\LNS$ is weakly successful but cannot be strengthened to an epistemic symmetric
strongly successful protocol.
Consider again the ``diamond'' graph of Section~\ref{subsec:diamond}, for which
we described various strongly successful strengthenings.
Also in ``static'' gossip $\LNS$ is weakly successful on this graph, since $01;30;20;31$ is successful.
All four possible first calls are symmetric.
After $21$, the remaining possible calls are $20$, $31$ and $30$.
But $20$ is bad, since 2 will never learn secret 3 that way.
Also $31$ is bad, since agent 1 will never learn the secret of 0.
The call $30$ is safe and in fact guarantees success, but by epistemic symmetry it cannot be allowed while $31$ is disallowed.
Therefore, in the static setting it is impossible to strengthen $\LNS$ on ``diamond'' such that it becomes strongly successful.
We expect a completely different picture for strengthening ``static'' gossip protocols in similar fashion as we did here,
for dynamic gossip.

We assumed synchronicity (a global clock) and common knowledge of the initial
gossip graph. These strong assumptions were made on purpose, because without
them agents will have even less information available and will therefore not be
able to coordinate any better. Such and other parameters for gossip problems are
discussed in~\cite{DGHHK2016:GossipParameters}. It is unclear what results still
can be obtained under fully distributed conditions, where agents only know their
own history of calls and their neighbors.

We wish to determine the logic of protocol-dependent knowledge $K^P_a$, and also
on fully distributed gossip protocols, without a global clock, and to further
generalize this beyond the setting of gossip.

\clearpage

\section*{Appendix: A Model Checker for Dynamic Gossip}
\addcontentsline{toc}{section}{Appendix: A Model Checker for Dynamic Gossip}

Analyzing examples of gossip graphs and their execution trees by hand is
tedious. To help us find and check the examples in this paper we wrote a
Haskell program which is available at \url{https://github.com/m4lvin/gossip}.
Our program can show and randomly generate gossip graphs, execute the protocols
we discussed and draw the resulting execution trees with epistemic edges.
The program also includes an epistemic model checker for the formal language
we introduced, similar to DEMO~\cite{JvE2007:DEMO}, but tailor-made for dynamic
gossip. For further details, see also \cite[Section~6.6]{GattingerThesis2018}.

Figure~\ref{figure:nExampleTreePart} is an example output of the implementation,
showing the execution tree for Example~\ref{example:nExample} up to two calls,
together with the epistemic edges for agent $2$, here called $c$. Note that
we use a more compact way to denote gossip graphs: lower case stands for a pure
number and capital letters for knowing the number and secret.

\begin{figure}[H]
  \centering
  \includegraphics[width=0.97\linewidth]{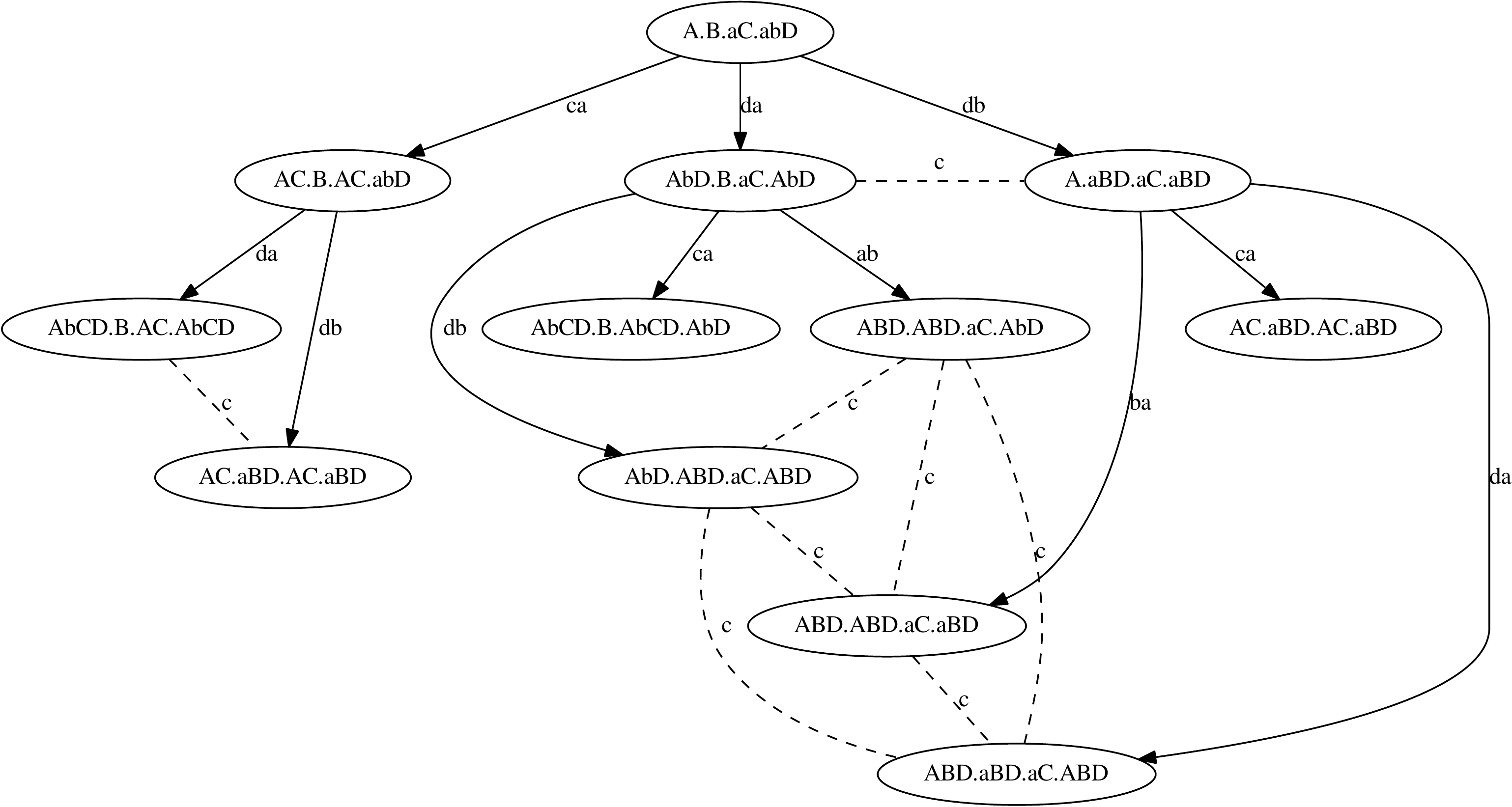}
  \caption{Two levels of the execution tree for Example~\ref{example:nExample},
    with epistemic edges for $c$.}\label{figure:nExampleTreePart}
\end{figure}

Our implementation can run different protocols on a given graph and output
a \LaTeX\ table showing and comparing the extension of those protocols.
Tables~\ref{table:nExampleExtensions} and~\ref{table:DiamondExampleExtensions}
have been generated in this way.
They provide details how various strengthenings behave on the gossip graphs
from Example~\ref{example:nExample} and Section~\ref{subsec:diamond}.

\begin{table}
  \centering
  \fontsize{8pt}{9.5pt}\selectfont
  \begin{tabular}{llllllllllllll}
  & $LNS$
  & $\cdot^\hard $
  & $\cdot^\soft$
  & $\cdot^\stephard$
  & $\cdot^{\stephard 2}$
  & $\cdot^{\stephard 3}$
  & $\cdot^{\stephard 4}$
  & $\cdot^\stepsoft$
  & $\cdot^{\stepsoft 2}$
  & $\cdot^{\stepsoft 3}$
  & $\cdot^{\stepsoft 4}$
  & $\cdot^{\stepsoft 5}$\\
$\epsilon$ &    &    &    &    &    &    & $\times$ &    &    &    &    &    \\
20 &    &    &    &    &    & $\times$ &    &    &    & $\times$ &    &    \\
20;30 &    &    &    &    & $\times$ &    &    &    & $\times$ &    &    &    \\
20;30;01 &    &    &    & $\times$ &    &    &    & $\times$ &    &    &    &    \\
20;30;01;31 & $\times$ &    &    &    &    &    &    &    &    &    &    &    \\
20;30;31 &    &    &    & $\times$ &    &    &    & $\times$ &    &    &    &    \\
20;30;31;01 & $\times$ &    &    &    &    &    &    &    &    &    &    &    \\
20;31 &    &    &    &    & $\times$ &    &    &    & $\times$ &    &    &    \\
20;31;10 &    &    &    & $\times$ &    &    &    & $\times$ &    &    &    &    \\
20;31;10;30 & $\times$ &    &    &    &    &    &    &    &    &    &    &    \\
20;31;30 &    &    &    & $\times$ &    &    &    & $\times$ &    &    &    &    \\
20;31;30;10 & $\times$ &    &    &    &    &    &    &    &    &    &    &    \\
30 &    & $\times$ &    &    &    & $\times$ &    &    &    &    &    &    \\
30;01 &    &    &    &    & $\times$ &    &    &    &    &    &    &    \\
30;01;20 &    &    &    & $\times$ &    &    &    &    &    &    &    &    \\
30;01;20;31 & $\times$ &    &    &    &    &    &    & $\times$ & $\times$ & $\times$ & $\times$ & $\times$ \\
30;01;31 &    &    &    & $\times$ &    &    &    & $\times$ & $\times$ & $\times$ & $\times$ & $\times$ \\
30;01;31;20 & $\times$ &    &    &    &    &    &    &    &    &    &    &    \\
30;20;01 &    &    &    &    & $\times$ &    &    &    &    &    &    &    \\
30;20;01;21;31 & $\checkmark$ &    & $\checkmark$ & $\checkmark$ &    &    &    & $\checkmark$ & $\checkmark$ & $\checkmark$ & $\checkmark$ & $\checkmark$ \\
30;20;01;31;21 & $\checkmark$ &    & $\checkmark$ &    &    &    &    & $\checkmark$ & $\checkmark$ & $\checkmark$ & $\checkmark$ & $\checkmark$ \\
30;20;21 &    &    &    &    & $\times$ &    &    &    &    &    &    &    \\
30;20;21;01;31 & $\checkmark$ &    & $\checkmark$ & $\checkmark$ &    &    &    & $\checkmark$ & $\checkmark$ & $\checkmark$ & $\checkmark$ & $\checkmark$ \\
30;20;21;31;01 & $\checkmark$ &    & $\checkmark$ &    &    &    &    & $\checkmark$ & $\checkmark$ & $\checkmark$ & $\checkmark$ & $\checkmark$ \\
30;20;31;01 &    &    &    & $\times$ &    &    &    &    &    &    &    &    \\
30;20;31;01;21 & $\times$ &    &    &    &    &    &    & $\times$ & $\times$ & $\times$ & $\times$ & $\times$ \\
30;20;31;21 &    &    &    & $\times$ &    &    &    &    &    &    &    &    \\
30;20;31;21;01 & $\times$ &    &    &    &    &    &    & $\times$ & $\times$ & $\times$ & $\times$ & $\times$ \\
30;31 &    &    &    &    & $\times$ &    &    &    &    &    &    &    \\
30;31;01 &    &    &    & $\times$ &    &    &    & $\times$ &    &    &    &    \\
30;31;01;20 & $\times$ &    &    &    &    &    &    &    &    &    &    &    \\
30;31;20 &    &    &    &    &    &    &    &    & $\times$ & $\times$ & $\times$ & $\times$ \\
30;31;20;01 &    &    &    & $\times$ &    &    &    & $\times$ &    &    &    &    \\
30;31;20;01;21 & $\times$ &    &    &    &    &    &    &    &    &    &    &    \\
30;31;20;21 &    &    &    & $\times$ &    &    &    & $\times$ &    &    &    &    \\
30;31;20;21;01 & $\times$ &    &    &    &    &    &    &    &    &    &    &    \\
31 &    &    &    &    &    & $\times$ &    &    &    &    &    &    \\
31;10 &    &    &    &    & $\times$ &    &    &    &    &    &    &    \\
31;10;20 &    &    &    & $\times$ &    &    &    & $\times$ & $\times$ & $\times$ & $\times$ & $\times$ \\
31;10;20;30 & $\times$ &    &    &    &    &    &    &    &    &    &    &    \\
31;10;30 &    &    &    & $\times$ &    &    &    & $\times$ &    &    &    &    \\
31;10;30;20 & $\times$ &    &    &    &    &    &    &    &    &    &    &    \\
31;20 &    &    &    &    & $\times$ &    &    &    & $\times$ & $\times$ & $\times$ & $\times$ \\
31;20;10 &    &    &    & $\times$ &    &    &    & $\times$ &    &    &    &    \\
31;20;10;30 & $\times$ &    &    &    &    &    &    &    &    &    &    &    \\
31;20;30 &    &    &    & $\times$ &    &    &    & $\times$ &    &    &    &    \\
31;20;30;10 & $\times$ &    &    &    &    &    &    &    &    &    &    &    \\
31;30 &    &    &    &    & $\times$ &    &    &    &    &    &    &    \\
31;30;10 &    &    &    & $\times$ &    &    &    & $\times$ &    &    &    &    \\
31;30;10;20 & $\times$ &    &    &    &    &    &    &    &    &    &    &    \\
31;30;20 &    &    &    & $\times$ &    &    &    & $\times$ & $\times$ & $\times$ & $\times$ & $\times$ \\
31;30;20;10 & $\times$ &    &    &    &    &    &    &    &    &    &    &    \\
\end{tabular}

  \caption{N Example~\ref{example:nExample}: Extensions of strengthenings.}\label{table:nExampleExtensions}
\end{table}

\begin{table}
  \centering
  \fontsize{9pt}{11.5pt}\selectfont
  \begin{tabular}{llllllllllllll}
  & $LNS$
  & $\cdot^\hard $
  & ${(\cdot^\hard)}^\stephard$
  & $\cdot^\soft$
  & $\cdot^\stephard$
  & $\cdot^{\stephard 2}$
  & $\cdot^{\stephard 3}$
  & $\cdot^{\stephard 4}$
  & $\cdot^\stepsoft$
  & $\cdot^{\stepsoft 2}$
  & $\cdot^{\stepsoft 3}$
  & ${(\cdot^\stepsoft)}^{\stephard 3}$\\
$\epsilon$ &    &    &    &    &    &    &    & $\times$ &    &    &    &    \\
01 &    &    &    &    &    & $\times$ &    &    &    &    &    &    \\
01;21 &    &    &    &    & $\times$ &    &    &    & $\times$ & $\times$ & $\times$ &    \\
01;21;30 & $\times$ &    &    &    &    &    &    &    &    &    &    &    \\
01;21;31 & $\times$ &    &    &    &    &    &    &    &    &    &    &    \\
01;30 &    &    &    &    & $\times$ &    &    &    &    &    &    &    \\
01;30;21 & $\times$ &    &    &    &    &    &    &    & $\times$ & $\times$ & $\times$ &    \\
01;31 &    &    &    &    & $\times$ &    &    &    &    &    &    &    \\
01;31;21 & $\times$ &    &    &    &    &    &    &    & $\times$ & $\times$ & $\times$ &    \\
21 &    &    &    &    &    & $\times$ &    &    &    &    &    &    \\
21;01 &    &    &    &    & $\times$ &    &    &    & $\times$ & $\times$ & $\times$ &    \\
21;01;30 & $\times$ &    &    &    &    &    &    &    &    &    &    &    \\
21;01;31 & $\times$ &    &    &    &    &    &    &    &    &    &    &    \\
21;30 &    &    &    &    &    &    &    &    &    & $\times$ & $\times$ &    \\
21;30;01 &    &    &    &    & $\times$ &    &    &    & $\times$ &    &    &    \\
21;30;01;31 & $\times$ &    &    &    &    &    &    &    &    &    &    &    \\
21;30;31 &    &    &    &    & $\times$ &    &    &    & $\times$ &    &    &    \\
21;30;31;01 & $\times$ &    &    &    &    &    &    &    &    &    &    &    \\
21;31 &    &    &    &    & $\times$ &    &    &    &    &    &    &    \\
21;31;01 & $\times$ &    &    &    &    &    &    &    & $\times$ & $\times$ & $\times$ &    \\
30 &    &    & $\times$ &    &    &    & $\times$ &    &    &    &    &    \\
30;01 &    & $\times$ &    &    &    & $\times$ &    &    &    &    &    &    \\
30;01;21;31 & $\checkmark$ &    &    & $\checkmark$ &    &    &    &    & $\checkmark$ & $\checkmark$ & $\checkmark$ &    \\
30;01;31;21 & $\checkmark$ &    &    & $\checkmark$ & $\checkmark$ &    &    &    & $\checkmark$ & $\checkmark$ & $\checkmark$ & $\checkmark$ \\
30;21;01 &    &    &    &    & $\times$ &    &    &    &    &    &    &    \\
30;21;01;31 & $\times$ &    &    & $\times$ &    &    &    &    & $\times$ & $\times$ & $\times$ &    \\
30;21;31 &    &    &    &    & $\times$ &    &    &    &    &    &    &    \\
30;21;31;01 & $\times$ &    &    & $\times$ &    &    &    &    & $\times$ & $\times$ & $\times$ &    \\
30;31 &    & $\times$ &    &    &    & $\times$ &    &    &    &    &    &    \\
30;31;01;21 & $\checkmark$ &    &    & $\checkmark$ & $\checkmark$ &    &    &    & $\checkmark$ & $\checkmark$ & $\checkmark$ & $\checkmark$ \\
30;31;21;01 & $\checkmark$ &    &    & $\checkmark$ &    &    &    &    & $\checkmark$ & $\checkmark$ & $\checkmark$ &    \\
31;01;21;30 & $\checkmark$ &    &    & $\checkmark$ &    &    &    &    & $\checkmark$ & $\checkmark$ & $\checkmark$ &    \\
31;01;30;21 & $\checkmark$ &    &    & $\checkmark$ & $\checkmark$ &    &    &    & $\checkmark$ & $\checkmark$ & $\checkmark$ &    \\
31;10;21;30 & $\checkmark$ &    &    & $\checkmark$ &    &    &    &    & $\checkmark$ & $\checkmark$ & $\checkmark$ &    \\
31;10;30;21 & $\checkmark$ & $\checkmark$ & $\checkmark$ & $\checkmark$ & $\checkmark$ & $\checkmark$ & $\checkmark$ &    & $\checkmark$ & $\checkmark$ & $\checkmark$ & $\checkmark$ \\
31;21;01;30 & $\checkmark$ &    &    & $\checkmark$ &    &    &    &    & $\checkmark$ & $\checkmark$ & $\checkmark$ &    \\
31;21;30 & $\checkmark$ &    &    & $\checkmark$ & $\checkmark$ &    &    &    & $\checkmark$ & $\checkmark$ & $\checkmark$ &    \\
31;30;10;21 & $\checkmark$ & $\checkmark$ & $\checkmark$ & $\checkmark$ & $\checkmark$ & $\checkmark$ & $\checkmark$ &    & $\checkmark$ & $\checkmark$ & $\checkmark$ & $\checkmark$ \\
31;30;21 & $\checkmark$ &    &    & $\checkmark$ &    &    &    &    & $\checkmark$ & $\checkmark$ & $\checkmark$ &    \\
\end{tabular}

  \caption{Diamond Example of Section~\ref{subsec:diamond}: Extensions of strengthenings,
  after $20$.}\label{table:DiamondExampleExtensions}
\end{table}

\clearpage

\interlinepenalty=10000 
\bibliographystyle{myplainurl}
\bibliography{ck-and-gossip}

\end{document}